\documentclass[aps,pra,showpacs,amsmath,amssymb,superscriptaddress,reprint,10pt,longbibliography]{revtex4-1}

\usepackage[english]{babel} 
\usepackage[T1]{fontenc} 
\usepackage[utf8]{inputenc}

\usepackage{times}

\usepackage{mathrsfs}
\usepackage{mathtools}
\usepackage{amsfonts,amstext}
\usepackage{amsmath,amsthm,amssymb}
\usepackage{dsfont}
\usepackage{enumerate,paralist}

\usepackage{graphicx}
\usepackage{xcolor}
\usepackage[colorlinks=true]{hyperref}
\hypersetup{pdfpagemode=UseNone,pdfstartview=} 
\usepackage{url}

\newcommand{\1}{\mathds{1}}
\newcommand{\id}{\1}
\newcommand{\RR}{\mathbb{R}}
\newcommand{\CC}{\mathbb{C}}
\newcommand{\ZZ}{\mathbb{Z}}

\newcommand{\SO}{\mathbb{SO}}

\newcommand{\tr}{\mathrm{tr}}
\newcommand{\e}{\ensuremath\mathrm{e}}
\newcommand{\X}{\sigma^x}
\newcommand{\Y}{\sigma^y}
\newcommand{\Z}{\sigma^z}

\newcommand{\J}{\mathcal{J}}

\newcommand{\W}{\mathcal{W}}
\newcommand{\F}{F_\W}
\renewcommand{\S}{\mathcal{S}_{n}}

\newtheorem{theorem}{Theorem}
\newtheorem{lemma}[theorem]{Lemma}
\newtheorem{proposition}[theorem]{Proposition}
\newtheorem{definition}[theorem]{Definition}
 \newtheorem{corollary}[theorem]{Corollary}

\newcommand{\argdot}{{\,\cdot\,}}

\newcommand{\norm}[1]{\left\Vert #1 \right\Vert} 
\newcommand{\mnorm}[1]{\norm{#1}_{\max{}}} 

\newcommand{\snorm}[1]{\norm{#1}_\infty} 

\newcommand{\ket}[1]{\left.\left|{#1}\right.\right\rangle}

\newcommand{\bra}[1]{\left.\left\langle{#1}\right.\right|}

\newcommand{\ketbra}[2]{\ket{#1} \!\! \bra{#2}}

\newcommand\vac{{\ketbra{\emptyset}{\emptyset}}}

\newcommand\vacket{{\ket{\emptyset}}}

\newcommand{\cm}{M}

\renewcommand{\t}{\top}

\renewcommand{\th}{\text{T}}

\renewcommand{\i}{\ensuremath\mathrm{i}}

\usepackage{tikz}
\definecolor{niceblue}{rgb}{0.33,0.5,0.8}%
\definecolor{yellowgreen}{RGB}{238,246,108}
\definecolor{yelloworange}{RGB}{253,184,19}
\usetikzlibrary{shapes.geometric, 
                fadings}
\usetikzlibrary{patterns}
\usetikzlibrary{calc,shapes.callouts,shapes.arrows}
  \definecolor{jens}{rgb}{1,1,1}
  
  \definecolor{leandro}{rgb}{.4,.0,.9}

  \definecolor{marek}{rgb}{.5,.5,.2}

\usepackage[normalem]{ulem}

\newcommand{\fu}
  {{Dahlem Center for Complex Quantum Systems, 
    Freie Universit\"{a}t Berlin, 
    Germany}}
\newcommand{\rio}
  {{Instituto de F\'isica, 
    Universidade Federal do Rio de Janeiro, 
 		   P. O. Box 68528, 
 		   Rio de Janeiro, 
		   RJ 21941-972, 
		   Brazil}}

\newcommand{\natal}
  {{International Institute of Physics, 
    Federal University of Rio Grande do Norte, 
59070-405 Natal,
	   Brazil}}

\newcommand{\saopaolo}{{
		  ICTP South American Institute for Fundamental Research, 
Instituto de F\'isica Te\'orica, UNESP-Universidade Estadual Paulista R. Dr. Bento T. Ferraz 271, Bl. II, S\~ao Paulo 01140-070, SP, Brazil 	}}

\newcommand{\ug}
	{Institute of Theoretical Physics and Astrophysics, 
	University of Gda\'{n}sk, 
	Poland}
  \newcommand{\hhu}{
	Institute for Theoretical Physics,
	Heinrich Heine University D{\"u}sseldorf, 
	Germany
}

\makeatletter 
 \hypersetup{pdftitle={Fidelity witnesses for certifying quantum simulations of lattice models},
	     pdfauthor={Marek Gluza, Martin Kliesch, Jens Eisert, Leandro Aolita},
	     pdfsubject={PACS numbers: },
	     pdfkeywords={Quantum simulation, verification, certification, fermionic gaussian state, fermionic gaussian operation, Gaussian fermions, fermionic gaussian evolution, importance sampling, dynamical quenches, Ising model, adiabatic state preparations, fidelity, fidelity witness, transverse field, Jordan-Wigner, fermionic linear optics, free fermions, certify quantum simulation, covariance matrix, fermion, fidelity lower bound, quantum fidelity, single-shot measurement protocol}
	    }
\makeatother

\begin{document}

\title{Fidelity witnesses for fermionic quantum simulations}

\author{M.\ Gluza}\affiliation{\fu}
\author{M.\ Kliesch}\affiliation{\ug}\affiliation{\hhu}
\author{J.\ Eisert}\affiliation{\fu}
\author{L.\ Aolita}\affiliation{\rio}\affiliation{\natal}\affiliation{\saopaolo}

\date{\today}

\begin{abstract}
The experimental interest and developments in quantum spin-1/2-chains has increased uninterruptedly over the last decade. In many instances, the target quantum simulation belongs to the broader class of non-interacting fermionic models, constituting an important benchmark. In spite of this class being analytically efficiently tractable, no direct certification tool has yet been reported for it. In fact, in experiments, certification has almost exclusively relied on notions of quantum state tomography scaling very unfavorably with the system size. Here, we develop experimentally-friendly fidelity witnesses for all pure fermionic Gaussian target states. Their expectation value yields a tight lower bound to the fidelity and can be measured efficiently. We derive witnesses in full generality in the Majorana-fermion representation and apply them to experimentally relevant spin-1/2 chains. 
Among others, we show how to efficiently certify strongly out-of-equilibrium dynamics in critical Ising chains. At the heart of the measurement scheme is a variant of importance sampling specially tailored to overlaps between covariance matrices. The method is shown to be robust against finite experimental-state infidelities.
\end{abstract}
\maketitle

\renewcommand{\i}{\ensuremath\mathrm{i}}
Quantum simulators are specific-purpose quantum devices that are able to efficiently simulate phenomena of interest thought to be
not directly accessible otherwise \cite{Fey82}. Already at scales of tens of particles they have the potential to 
outperform today's most powerful supercomputers and help us explain unclear physical effects, as well as give boosts in crucial technological areas \cite{cirac2012goals}. 
In addition, they constitute an intermediate milestone towards the ultimate goal of realizing 
large-scale universal quantum computers. This has fuelled  impressive experimental advances in multiple quantum technologies \cite{aspuru2012photonic,blatt2012quantum,houck2012chip_review,bloch2012quantum,Schneider12,EisFriGog15}.
A type of quantum many-body systems to whom experimental simulations have devoted considerable efforts over the last decade 
are given by one-dimensional (1D) lattices of interacting spin-1/2 particles, or \emph{spin-1/2 chains}, for short. In particular, 
even though they call into the efficiently classically simulable regime, the well-known transverse-field (TF) Ising and XY models have risen
to constitute important basic testbeds for the most advances experimental simulations, e.g. with ion-trap \cite{Friedenauer08,Kim10,Islam11,lanyon2016efficient}, superconducting-circuit \cite{barends2016digitized} and circuit quantum electrodynamics \cite{Yves} platforms.

At least two facts justify the significant interest in these specific models. The first one is that they display a vast physical richness: For instance, the TF Ising model---which is, actually, a subclass of the TF XY model---features a quantum phase transition \cite{Pfeuty70,sachdev2007quantum,PhysRevX.4.031008,francesco2012conformal} as well as topologically and spectrally interesting effects \cite{kitaev2001unpaired,you2014symmetry,PhysRevB.92.115137,pollmann2012symmetry,grimm2002spectrum}, and is relevant for quantum speed-ups in certain optimization problems \cite{nishimori2016exponential}. The second one is that, for nearest-neighbor interactions, they can be analytically solved, e.g. by mapping them into systems of free---i.e., non-interacting---fermions \cite{Pfeuty70}.
This allows for in-depth theoretical studies of their dynamics \cite{calabrese2012quantum,calabrese2012quantumII,calabrese2004entanglement,dziarmaga2005dynamics,Gaussification,RoschTransport}. 
From a broader perspective, these models belong to a more general class of exactly solvable systems known as \emph{non-interacting quantum systems},
also referred to as \emph{fermionic linear optics} \cite{knill2000efficient,knill2001fermionic,terhal2002classical,bravyi2004lagrangian,kitaev2006anyons,Melo13,bravyicapacity}. This class is the fermionic counterpart of the Gaussian formalism for bosons \cite{Ferra05, Weed12}, which plays a major role in quantum information and quantum optics. It includes, e.g., tight-binding models important in condensed-matter physics, certain interacting bosonic chains that can be fermionized \cite{haldane1981luttinger,vonDelft_bosonization,fradkin2013field}, and spin-1/2 systems in 2-dimensional lattices, such as the celebrated Kitaev's honeycomb model \cite{kitaev2006anyons}, which exhibits non-Abelian excitations.

Unfortunately, the exact analytical solution of a model does not imply that one can efficiently \emph{certify} the correctness of an uncharacterised experimental simulation of it. Furthermore, even if the computational complexity of the target simulation is low, the number of measurements required for its certification can be exponentially high in the lattice size without the adequate certification method. This is the case, e.g., for full state tomography (FST). Certification tools not relying on FST exist \cite{Gross10,ExperimentalCS,Toth10,cramer2010efficient,flammia2011direct,SilLanPou11,Flammia12,aolita2015reliable,steffens2015towards,Hangleiter16}, each one efficient on a different subclass of simulations. However, none of these can efficiently handle fermionic linear optics. In fact, almost all  \cite{Friedenauer08,Kim10,Islam11,Yves,barends2016digitized} the above-mentioned experiments relied on FST. The simulation of Ref.~\cite{lanyon2016efficient}, in contrast, was certified with matrix-product state tomography \cite{cramer2010efficient}. 
This is a powerful method that covers a broad class of chains but tolerates  little long-range entanglement, so that non-trivial evolutions are in practice tractable only over short times \cite{cramer2010efficient, lanyon2016efficient}. Indeed, generic spin chains out of the equilibrium \cite{eisert2006general}, or even very natural, static free-fermionic states \cite{ramirez2015entanglement,eisler2016entanglement}, involve large amounts of entanglement along the lattice.
Today, a major roadblock for further experimental progress in spin-chain simulations (and in many-body quantum technologies in general) is their certification. 

Here, we develop efficient \emph{fidelity witnesses} for all pure fermionic Gaussian target states. 
These are experimentally-friendly observables whose expectation value (on an arbitrary experimental state) yields a tight lower bound to the fidelity with the target. Hence, they allow for \emph{unconditional} certification, i.e., without any a-priori knowledge of the experimental setup or imperfections. 
We derive the witnesses in full generality in the Majorana-fermion representation, and then apply them to experimentally relevant spin-1/2 chains as examples. 
Among others, we show how to efficiently certify any \emph{sudden quench} (i.e., strongly out-of-equilibrium dynamics) in a critical TF Ising chain with nearest-neighbor interactions. 
The measurement scheme relies on a new variant of \emph{importance sampling} tailored to overlaps between covariance matrices, which is potentially interesting on its own.
As a result, the number of measurements required for the certification only has a modest scaling with the lattice size, i.e., a small \emph{sample complexity}, for which we present upper bounds.
Moreover, the method is robust against finite experimental-state infidelities, in the sense of there always existing a closed ball of valid states that are correctly accepted by the certification test. 
Finally, we provide also a totally general construction, not restricted to fermions or Gaussian states, of (possibly non-efficient) fidelity witnesses for arbitrary pure target states, which may also be useful in other scenarios.

\emph{Preliminaries}. 
Consider a system of $L$ spin-less fermionic atoms, from now on referred to as fermionic modes, with creation and annihilation operators $f_j^\dagger$ and $f_j$, respectively, for $j=1,2,\ldots,L$, satisfying the canonical anti-commutation relations $\{f_j,f_k^\dagger\}\equiv f_j  f_k^\dagger+ f_k^\dagger f_j = \delta_{j,k}$ and $\{f_j,f_k\}=\{f_j^\dagger,f_k^\dagger\}=0$, with $\delta_{j,k}$ the Kronecker symbol. 
Let us next introduce the self-adjoint Majorana mode operators
\begin{equation}
\label{def:majorana_mode_ops}
m_{2j-1}\coloneqq f_j+f_j^\dagger\, ,\qquad 
m_{2j}   \coloneqq-\i\, (f_j-f_j^\dagger),
\end{equation}
with anti-commutation relations $\{m_j,m_k\}=2\,\delta_{j,k}$. 
We say that the fermionic system is \emph{free}, \emph{Gaussian}, or \emph{linear-optical} \cite{knill2000efficient,knill2001fermionic,terhal2002classical,bravyi2004lagrangian,kitaev2006anyons,Melo13,bravyicapacity}, 
if it is governed by a quadratic Hamiltonian
\begin{equation}
\label{eq:H}
H=\frac{\i}{4}\sum_{j,k=1}^{2L} A_{j,k} \, m_j\,  m_k,
\end{equation}
where $\boldsymbol{A}=-\boldsymbol{A}^\t\in\RR^{2L\times 2L}$ is called the coupling matrix. 

The term ``free'' or ``non-interacting'' stems from the fact that $H$ is unitarily equivalent to a Hamiltonian of $L$ fermions not featuring any off-diagonal
couplings. In the bosonic realm, this is the defining property of Gaussian systems \cite{Ferra05, Weed12}, which justifies the term ``Gaussian''. In turn, what is linear about ``fermionic linear-optics'' is the time evolution of the mode operators in the Heisenberg picture,
\begin{align}
\label{eq:mode_transform}
m_j(t)\coloneqq U^\dagger(t)\, m_j\,U(t) =\sum_{k=1}^{2L}Q_{j,k}(t)\, m_k ,
\end{align}
where $U(t)\coloneqq \e^{-\i Ht}$, for $t\in\RR$, is a \emph{fermionic Gaussian unitary} and $\boldsymbol{Q}(t) \coloneqq e^{t \boldsymbol{A}}\in \SO(2L)$ its 
representation in mode space\cite{kraus2010generalized}, see Appendix~\ref{app:FLO} for a simple derivation.

Finally, it is useful to introduce, for any state $\varrho$ (Gaussian or not), the real anti-symmetric covariance matrix $\boldsymbol{\cm}(\varrho)$ 
with elements
\begin{equation}
\label{eq:def_cm}
 \cm_{j,k}(\varrho)\coloneqq\frac{\i}{2}\tr\bigl([m_j,m_k]\,\varrho\bigr) .
\end{equation}
This matrix 
contains the expectation values of 
the single-mode densities $\langle {n}_j \rangle\coloneqq \langle f^\dagger_j\, f_j\rangle$ as well as the two-mode currents $\langle f^\dagger_j\,f_k+\mathrm{h.c.} \rangle$ and pairing terms $\langle f^\dagger_j\,f^\dagger_k+\mathrm{h.c.} \rangle$. 

\begin{figure}[t!]
\centering
\includegraphics[width=.65\linewidth]{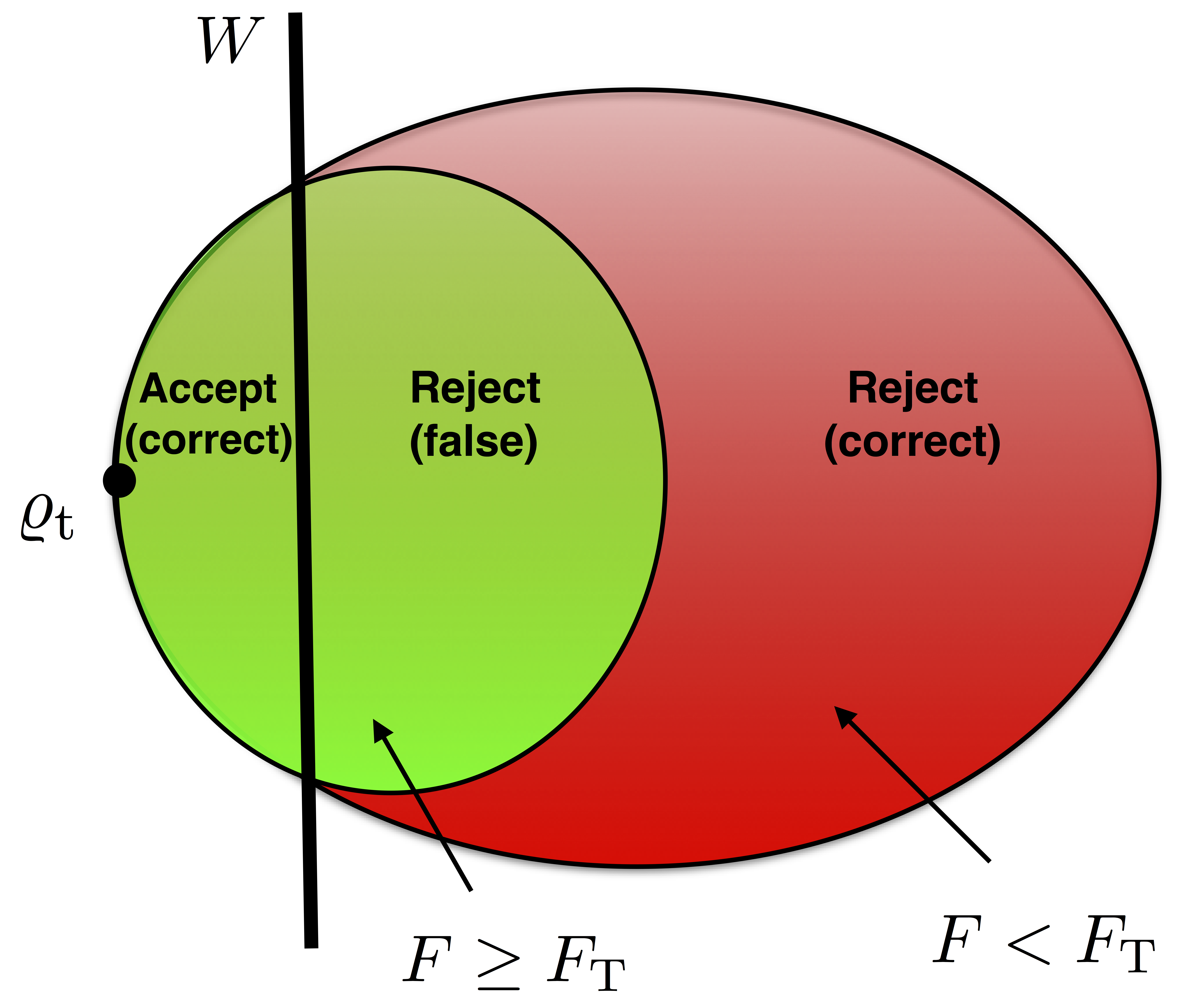}
\caption{
	\textbf{Geometrical representation of a fidelity witness}. A pure target state $\varrho_\text{t}$ lies at the boundary of state space. For any  fixed fidelity threshold $F_\text{T}$, the valid experimental states are defined by $F\geq F_\text{T}$ (green). The states with $F< F_\text{T}$ are invalid (red). A fidelity witness $\W$ defines a hyper-plane (straight  line), to the left of which only valid states are found and to the right of which both valid as well as invalid ones are found. The certification test consists of accepting all states on the left and rejecting all those on the right. Hence, a significant subset of valid states is sacrificed, as in weak-membership problems. However, in return, the experimental estimation is considerably more efficient than in schemes attempting to separate the valid from the invalid states (strong-membership problems).
}
\label{fig:certification}
\end{figure}
\emph{Fidelity witnesses}.  We consider throughout a (known) pure target state $\varrho_\text{t}$ and an arbitrary, unknown experimental preparation $\varrho_p$. 
Their closeness is measured by their \emph{fidelity}
\begin{equation}
  \label{eq:fidelity_def}
  F\coloneqq F(\varrho_\text{t},\varrho_p):=\tr\big[(\sqrt{\varrho_\text{t}}\,\varrho_p^{\dagger}\,\sqrt{\varrho_\text{t}})^{1/2}\big]^2=\tr\big[\varrho_\text{t}\,\varrho_p\big],
\end{equation}
where the last equality holds because $\varrho_\text{t}$ is pure. With this, the pivotal notion of our work can be defined:
\begin{definition}[Fidelity witnesses]
\label{def:witness}
An observable $\W$ is a \emph{fidelity witness} for $\varrho_\text{t}$ if, for  $\F(\varrho_\text{p})\coloneqq\tr[\W\, \varrho_p]$, it holds that
\begin{compactenum}[i)]
 \item \label{item:iff} $\F(\varrho_\text{p})=1$ if, and only if, $\varrho_p=\varrho_\text{t}$, and
 \item \label{item:geq} $\F(\varrho_\text{p}) \leq F$ for all states $\varrho_p$.
\end{compactenum}
\end{definition}
The term ``witness'' refers to the property that, for any fixed threshold $F_\text{T}$, finding $\F(\varrho_\text{p})\geq F_\text{T}$
witnesses that $F\geq F_\text{T}$; but if $\F(\varrho_\text{p})<F_\text{T}$ is found, then nothing can be said about $F$  (see Fig. \ref{fig:certification}). This is the least information about $\varrho_p$ needed to certify its fidelity with $\varrho_\text{t}$. The situation is reminiscent of entanglement witnesses \cite{guehne09}, which detect some entangled states and discard all non-entangled ones. The difference is that fidelity witnesses explicitly realise the extremality-based intuition of ``corralling valid states against the boundary''. 
Specific witnesses have been built for ground states of local Hamiltonians \cite{cramer2010efficient,vanderNest,Hangleiter16} and Gaussian as well as non-Gaussian output states of bosonic linear-optical circuits \cite{aolita2015reliable}. In Appendix \ref{sec:Proof_general_wti_construction}, we present (possibly non-efficient) fidelity witnesses of arbitrary target states with no assumption other than they being pure. A special case of such generic construction is the following (efficient) witnesses for the free-fermionic setting.

Any $L$-mode pure fermionic Gaussian target state $\varrho_\text{t}$ can be written as  
\begin{align}
\label{eq:taregt_state_def}
\varrho_\text{t}\coloneqq\ketbra{\psi_\text{t}}{\psi_\text{t}}\qquad\text{with}\qquad\ket{\psi_\text{t}}\coloneqq U\ket{\boldsymbol{\omega}},
\end{align}
for a fermionic Gaussian unitary $U$, as defined below Eq.\ \eqref{eq:mode_transform}, where $\boldsymbol{\omega}\coloneqq (\omega_1, \ldots,\omega_L)$ is any $L$-bit string. The ket $\ket{\boldsymbol{\omega}}$ represents the Fock-basis state vector with $\omega_j$ ($=$ 0 or 1) excitations in mode $j$, i.e., $n_j\ket{\boldsymbol{\omega}}=\omega_j \ket{\boldsymbol{\omega}}$, for $j=1, \ldots , L$, and $n_j\coloneqq f^\dagger_j\, f_j$.
It is also convenient to introduce $n^{(\boldsymbol{\omega})}\coloneqq\sum_{j=1}^L\left[ (1-\omega_j)  n_j +\omega_j(\id- n_j)\right ]$, the total fermion-number operator in the locally-flipped basis in which $\boldsymbol{\omega}$ is the is the null string, i.e. $n^{(\boldsymbol{\omega})}\ket{\boldsymbol{\omega}}=0$. In other words, $\ket{\psi_\text{t}}$ represents the so-called Fermi-sea state and the eigenstates of $n^{(\boldsymbol{\omega})}$ its excitations.
In Appendix \ref{sec:Proof_general_wti_construction}, we show that the observable
\begin{equation} 
\label{eq:def_W}
\W = U\Big(\1- n^{(\boldsymbol{\omega})}\Big) U^\dagger
\end{equation}
is a fidelity witness  for $\varrho_\text{t}$. 
Expression \eqref{eq:def_W} is the fermionic analogue of the bosonic Gausssian-state witnesses of Ref.\ \cite{aolita2015reliable}, with a crucial difference: While for bosons only the Fock-basis state vector $\ket{\boldsymbol{0}}$ is Gaussian, for fermions all $2^L$ Fock-basis vectors $\ket{\boldsymbol{\nu}}$ are Gaussian as they satisfy Wick's theorem \cite{bravyi2004lagrangian}. In fact, for mixed states, all single-mode
states are Gaussian, in sharp contrast to the bosonic case.

\emph{Measurement scheme}.  
Taking the expectation value of Eq.\ \eqref{eq:def_W} with state $\varrho_\text{p}$ yields (see Appendix~\ref{sec:evaluation})
\begin{equation}
\label{eq:def_Fgeq}
  \F(\varrho_\text{p}) =
  1 + \frac14 \tr\Bigl[\big(\boldsymbol{M}(\varrho_\text{p})-\boldsymbol{M}(\varrho_\text{t})\big)^\t \boldsymbol{M}(\varrho_\text{t}) \Bigr] \, ,
\end{equation}
where $\boldsymbol{M}(\varrho_\text{p})$ and $\boldsymbol{M}(\varrho_\text{t})$ are the covariance matrices of  $\varrho_\text{t}$ and $\varrho_\text{p}$, respectively. This expression holds also for bosonic Gaussian witnesses \cite{aolita2015reliable} and turns out very useful for the measurement of $\F(\varrho_\text{p})$. 
We call 
$\Omega\coloneqq\{ (j,k)\,:\cm_{j,k}(\varrho_\text{t})\neq0, \text{ for }1\le j<k\le 2L\}$ the set of non-zero entries of $\boldsymbol{M}(\varrho_\text{t})$. 
Then Eqs.~\eqref{eq:def_cm} and \eqref{eq:def_Fgeq} imply that if one measures on $\varrho_\text{p}$ all $|\Omega|\le2L^2+L$ observables 
$\i  [m_{j},m_{k}] /2$ with indices in $\Omega$, then one can estimate $\F(\varrho_\text{p})$. 
However, this is not the most efficient procedure (see Appendices~\ref{sec:single_shot_imp_sampling} and \ref{sec:entrywise_eval}).

A more efficient approach is to exploit \emph{importance sampling} techniques, where a subset of the $|\Omega|$ observables  is randomly selected for measurement according to its importance for $\W$. These techniques have been applied in Hilbert space to the estimation of state overlaps, where they yield efficient schemes only for a specific type of target states \cite{flammia2011direct,SilLanPou11}. 
Here we apply them in mode space to efficiently estimate overlaps between fully general covariance matrices. The starting point is to identify a random variable $X$ and an importance distribution $P\coloneqq\{P_\mu\}_\mu$, with $X$ taking the value $X_\mu$ with probability $P_\mu$, such that $\tr\bigl[\boldsymbol{M}(\varrho_\text{p})^\t\,\boldsymbol{M}(\varrho_\text{t})\bigr]$ is expressed as the mean value of $X$, i.e., 
\begin{align}
\mathbb E[X]= \sum_{\mu} P_\mu\, X_\mu = \tr\bigl[\boldsymbol{M}(\varrho_\text{p})^\t\,\boldsymbol{M}(\varrho_\text{t})\bigr]\;.
\label{eq:fundamental}
\end{align}
Then, if one can experimentally sample $X$ from $P$, $\mathbb E[X]$ can be approximated  by the finite-sample average $\mathcal X^*\coloneqq \sum_{m=1}^\mathcal{N} X_{\mu(m)}/{\mathcal{N}}$, where $X_{\mu(m)}$ is the value of $X$ at the $m$-th experimental run and $\mathcal{N}$ is the total sample size (number of runs). 
Next, we present a choice of $X$ and $P$ particularly suited to estimate $\F(\varrho_\text{p})$.

To this end, let us first define $\hat  m^{(\beta)}_{j,k}$ as the projector onto the eigenstate of the observable $\i  m_{j}m_{k}$ with eigenvalue $\beta=\pm1$, for $(j,k)\in\Omega$. Then, identifying $\mu$ with the triple $(\beta, j,k)$ and using the short-hand notation 
\begin{equation}
\label{eq:short_hand_notation}
|\boldsymbol{M}(\varrho_\text{t})| \coloneqq \sum_{(j,k) \in \Omega} |\cm_{j,k}(\varrho_\text{t})|\le 2L^2,
\end{equation}
we choose 
\begin{align}
\label{eq:choice_X}
X_{\beta, j,k}\coloneqq 2\, |\boldsymbol{M}(\varrho_\text{t})|\,\beta\,\text{sgn}\big[\cm_{j,k}(\varrho_\text{t})\big]
\end{align}
and 
\begin{align}
\label{eq:choice_P}
 P_{\beta,j,k}\coloneqq\frac{\tr\Big[\hat  m^{(\beta)}_{j,k}\,\varrho_{p}\Big]  |\cm_{j,k}(\varrho_\text{t})|}{|\boldsymbol{M}(\varrho_\text{t})|} \, .
\end{align}
This choice satisfies Eq.\ \eqref{eq:fundamental}, as explicitly shown in Appendix~\ref{sec:single_shot_imp_sampling}.
In the experiment, in turn, for each run, one chooses $(j,k)$ according to 
$P_{j,k}\coloneqq{|\cm_{j,k}(\varrho_\text{t})|}/{|\boldsymbol{M}(\varrho_\text{t})|}$ and measures $\i  m_{j}m_{k}$ on $\varrho_\text{p}$, which outputs $\beta$ with probability $P_{\beta\vert j,k}\coloneqq \tr[\hat  m^{(\beta)}_{j,k}\,\varrho_{p}]$. Substituting the obtained $(j,k)$ and $\beta$ in Eq.\ \eqref{eq:choice_X}, one samples $X_{\beta, j,k}$ with probability  $ P_{\beta,j,k}$, as desired. 
As for the experimental accessibility of the observables, for the relevant case of spin-1/2 chains each $\i m_jm_k$ corresponds to a product of Pauli matrices, as discussed below.

This single-shot importance-sampling approach  does not necessarily yield a good estimate of each individual entry of $\boldsymbol{M}(\varrho_\text{p})$, as unlikely observables according to $P_{j,k}$ are measured seldomly. The method  is specially tailored to directly obtain $\F(\varrho_\text{p})$. In fact, the resulting estimate $\mathcal X^*$ yields an excellent approximation of $\tr\bigl[\boldsymbol{M}(\varrho_\text{p})^\t\,\boldsymbol{M}(\varrho_\text{t})\bigr]$ (in a formal sense given by 
Theorem \ref{thm:main} below), with which the right-hand side of Eq.\ \eqref{eq:def_Fgeq} can be immediately evaluated. This gives our final finite-sample estimate $\F^*(\varrho_\text{p})$ of $\F(\varrho_\text{p})$.

\emph{Sample complexity}.  The scaling in $L$ of the minimum (over all estimation strategies) number $\mathcal N_{\epsilon,\delta}(\mathcal W)$ of measurement runs required to estimate $\F(\varrho_\text{p})$, up to statistical error at most $\epsilon$ and with failure probability at most $\delta$, i.e., such that 
\begin{align}
\label{eq:large_dev_bound}
  \mathbb{P} \left (|\F(\varrho_\text{p}) - \F^*(\varrho_\text{p})| \le \epsilon \right ) \ge 1- \delta\;,
  \end{align}
for all $\varrho_\text{p}$,  is called the \emph{sample complexity} \cite{Flammia12,GogKliAol13,aolita2015reliable} of estimating $\F(\varrho_\text{p})$. In Appendix~\ref{sec:single_shot_imp_sampling} we compute the number of runs required with the measurement scheme described above, which sets the following upper bound on $\mathcal N_{\epsilon,\delta}(\mathcal W)$.
\begin{theorem}[Sample complexity of $\F$]
\label{thm:main}
Let $\epsilon >0$, $\delta\in(0,1)$, $\varrho_\text{t}$ given by Eq.~\eqref{eq:taregt_state_def}, and $\W$ by Eq.~\eqref{eq:def_W}. Then 
 \begin{align}
\label{eq:unified_exp_sample_complex}
\mathcal N_{\epsilon,\delta}(\W) \leq\left\lceil\frac{\ln(2/\delta) |\boldsymbol{M}(\varrho_\text{t})| ^2}{2\,\epsilon^2}\right\rceil.
\end{align}
\end{theorem}

Eq.~\eqref{eq:short_hand_notation}
implies that the right-hand side of Eq.~\eqref{eq:unified_exp_sample_complex} is never larger than 
$\left\lceil{2\ln(2/\delta) L^4}/{\epsilon^2}\right\rceil$. The scaling is thus polynomial in $L$ for all $\varrho_\text{t}$, which means that the scheme is efficient in the lattice size. Furthermore, for the physically-relevant case of $\varrho_\text{t}$ being the unique ground state of a local gapped Hamiltonian, the correlations $\tr\bigl([m_j,m_k]\,\varrho_\text{t}\bigr)$ decay exponentially with $|j-k|$ \cite{hastings2006spectral}. Then, $|\boldsymbol{M}(\varrho_\text{t})|\sim L \log(L)$, which leads to $\mathcal N_{\epsilon,\delta}(\W)\leq O\left(L^2\log^2(L) \right)$.

Finally, in Appendix~\ref{sec:entrywise_eval}, we study also a measurement scheme without importance sampling (where all $|\Omega|$ observables are measured) but exploiting the fact that all commuting observables with indices in $\Omega$ can be measured simultaneously in each run. This gives the bound $\mathcal N_{\epsilon,\delta}(\W)\leq  O\left({2\ln(2\,|\Omega|/\delta) L^4}/{\epsilon^2}\right)$, which, since $|\Omega|\le2L^2+L$, scales logarithmically worse in $L$ than in Eq.\ \eqref{eq:unified_exp_sample_complex}. We suspect that the bound in Eq.\ \eqref{eq:unified_exp_sample_complex} is close to being tight.

\emph{Spin-1/2 chains}.  
We denote a local spin operator acting at site $k$ by $\sigma^\alpha_k = \id_2^{{\otimes}(k-1)}\otimes \sigma^{\alpha}\otimes\id_2^{\otimes(L-k)}$ where $\sigma^\alpha$ for $\alpha=x,y,z$ are the Pauli matrices and $\id_2$ is the single-qubit identity.
Via the Jordan-Wigner transformation \cite{JordanWigner,LiebSchultzMattis}
\begin{align}
\label{eq:JW_transform}
\begin{split}
  {m}_{2k-1}= 
  \bigl(\prod_{j<k}\Z_j\bigr)\,\X_k,  \quad
  {m}_{2k}= 
  \bigl(\prod_{j<k}\Z_j\bigr)\,\Y_k,
  \end{split}
\end{align}
 the Hamiltonian in Eq.\ \eqref{eq:H} is equivalent \cite{Pfeuty70} to the experimentally-relevant  \cite{Friedenauer08,Kim10,Islam11,lanyon2016efficient,Yves,barends2016digitized} spin-1/2 Hamiltonian 
\begin{equation}
\label{eq:spin_H}
H_\text{spin}
=
-\sum_{k=1}^{L-1}  (J_k^x \,\X_j\,\X_{k+1}+J_k^y \,\Y_k\,\Y_{k+1})-\,\sum_{k=1}^L B_k\, \Z_k\, ,
\end{equation}
where $J_k^x,J_k^y\in\RR$ and $B_k\in\RR$ are respectively constant coupling and transverse-field strengths. 
Since these spin-1/2 chains are equivalent to free-fermionic systems for all parameter regimes, certifying quantum simulations of, e.g., adiabatic ground state preparations as well as sudden quenches amounts to certifying pure fermionic Gaussian states, as described above. 
Finally, note that Eqs. \eqref{eq:JW_transform} map each $\i m_jm_k$ to a product of Pauli matrices, as anticipated in the measurement scheme above.

\emph{Sudden quenches in critical Ising chains}.  
The 1D nearest-neighbor TF Ising Hamiltonian is given by Eq.\ \eqref{eq:spin_H} with $J_k^x= J$, $J^y= 0$, and $B_k= B$, for all $k=1, \ldots , L$, where $J,B>0$.
In a typical quench, the initial ground state $\ket{\psi(0)}\coloneqq
\ket{\uparrow}^{\otimes L}$ at a non-critical regime $J=0<B$, where $\ket{\uparrow}$ is an eigenvector of $\sigma^z$,
is evolved under the critical regime $J=B$, so as to generate a strong out-of-equilibrium evolution.
These quenches are particularly challenging to certify \cite{cramer2010efficient, lanyon2016efficient} because the time-evolved state vector $\ket{\psi(t)}$  rapidly acquires large amounts of entanglement.
Let us consider the simulation of such a quench by a \emph{digital quantum simulator}, which approximates the continuous time evolution with a Trotter-Suzuki pulse sequence $U(t)=e^{-\i\, t\,(H_B+H_J)}\approx U_T\coloneqq\left(e^{-\i\,\Delta t\,H_B}\,e^{-\i\,\Delta t\,H_J}\right)^T$, where $t=T\,\Delta t$ and $H_B$ and $H_J$ are the Ising Hamiltonians for $J=0$ and $B=0$, respectively.
The target covariance matrix is then $\boldsymbol{\cm}(\varrho_\text{t})=\boldsymbol{Q}(t)\, \boldsymbol{\cm}\big(\ketbra{\uparrow}{\uparrow}^{\otimes L}\big)\, \boldsymbol{Q}(t)^\t$, 
where $\boldsymbol{Q}(t)=e^{t\, \boldsymbol{A}(J,B)}$, with $\boldsymbol{A}(J,B)$ the coupling matrix of $H_B+H_J$, is the mode representation of the target time propagator $U(t)$ and 
\begin{align}
	\boldsymbol{\cm}\big(\ketbra{\uparrow}{\uparrow}^{\otimes L}\big)\coloneqq\oplus_{j=1}^L\begin{psmallmatrix} 
  0 & -1 \\ 1 & 0 
  \end{psmallmatrix}.
  \end{align}
In turn, the preparation's covariance matrix is given by $\boldsymbol{\cm}(\varrho_\text{p})=\boldsymbol{Q}_T\, \boldsymbol{\cm}\big(\ketbra{\uparrow}{\uparrow}^{\otimes L}\big)\, \boldsymbol{Q}_T^\t$, where $\boldsymbol{Q}_T=\left(e^{\Delta t\,\boldsymbol{A}(J)}\,e^{\Delta t\,\boldsymbol{A}(B)}\right)^T$, with $\boldsymbol{A}(J)$ ($\boldsymbol{A}(B)$) the coupling matrix of $H_B$ ($H_J$), corresponds to the discrete-time experimental evolution $U_T$,
see Fig.~\ref{fig:trotter} and Appendix~\ref{app:FLO} for technical details. 

\begin{figure}
 \includegraphics[trim = 0.85cm 0 0 0 , clip , width=.8\linewidth]{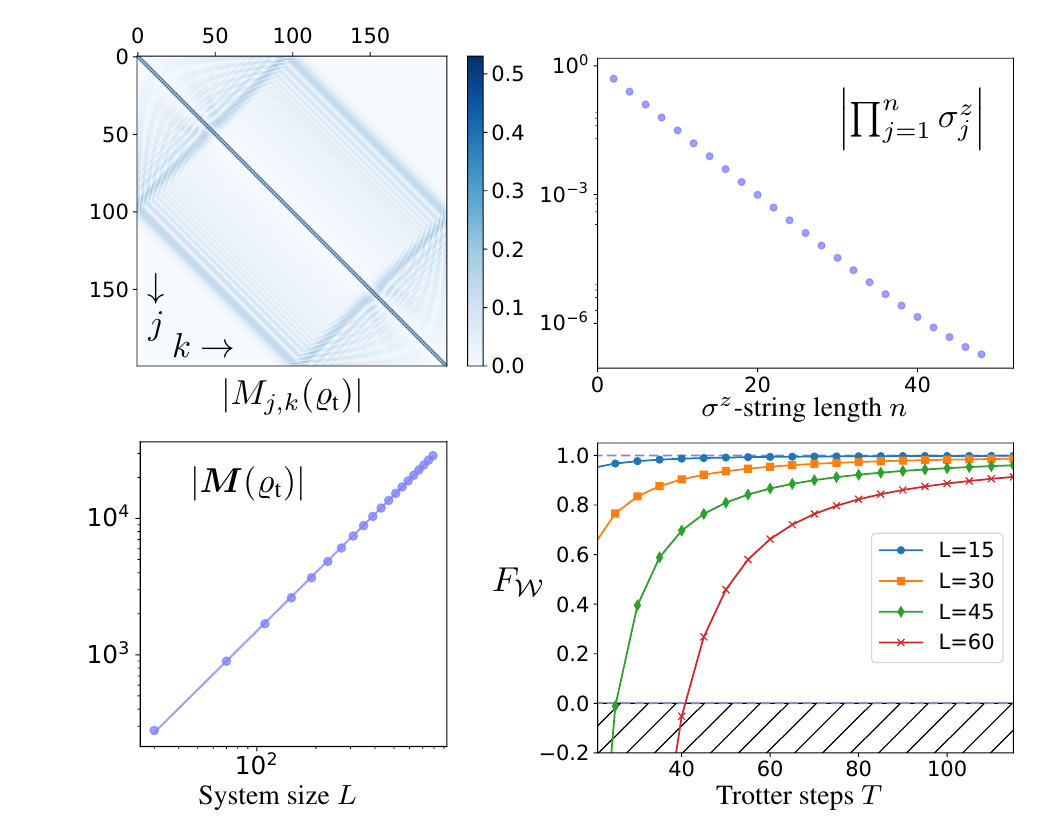}
\caption{
\textbf{Certification of a sudden quench in a critical spin-1/2 chain.} The target state vector is given by $\ket \uparrow ^{\otimes L}$ evolved, under the TF Ising Hamiltonian with $J=1=B$, up to time $t=L/8$, where $L$ is the number of spins. Top-left: Absolute values of the $4\,L^2$ entries of the target covariance matrix for $L=100$. The plot shows the correlation wavefront propagating parallel to the diagonal. At $t=L/8$, the wavefront has explored half the lattice size, in the sense that long-range correlations between spins of lattice distance $L/2$ have developed.
Top-right: An extensive amount of Pauli-matrix products (only one of which is shown) have exponentially small expectation values on the target state. This renders importance sampling in Hilbert space \cite{flammia2011direct,SilLanPou11} inapplicable. 
Importance sampling in mode space, in contrast, is efficient for estimating overlaps between arbitrary covariance matrices (see text). 
Bottom-left: $|\boldsymbol{M}(\varrho_\text{t})|$ as a function of $L$, with a power fit yielding the scaling $|\boldsymbol{M}(\varrho_\text{t})|\approx 2.11\times L^{1.42}$, so that the sample complexity is bounded by $\mathcal N_{\epsilon,\delta}(\mathcal W)\lesssim O(L^{2.84})$. 
Bottom-right: Expectation value of the fidelity witness, as a function of $L$, on a preparation of the continuos-time evolved state by a digital quantum simulator with $T$ Trotter-Suzuki pulses. 
As $L$ increases, $T$ needs to increase to keep a value of the fidelity lower bound constant.
}
\label{fig:trotter}
\end{figure}

\emph{Discussion}.  We have shown how to certify experimental states of dimension $2^{2\,L}$ with at most $O(L^4)$ measurements, with no assumption whatsoever on the experimental imperfections,  for all pure fermionic Gaussian target states. Moreover, for targets given by ground states of gapped free-fermionic Hamiltonians, the number of experimental repetitions reduces to $O\left(L^2\log^2(L) \right)$. In addition, in Appendix~\ref{sec:robustness} we prove that there always exists a closed ball of valid states that are correctly accepted by the certification test, so that the test is robust against finite experimental deviations.

Our results are directly relevant to recent experiments with spin chains \cite{BlattRoos12,Schneider12,Friedenauer08,Kim10,Islam11,Yves,barends2016digitized} as well as potential implementations of Kitaev's honeycomb model  \cite{schmied2011quantum,mielenz2015freely}.
In real-life digital simulations
\cite{Islam11,Yves,barends2016digitized}, apart from the Trotterisation errors, also heating and noise will of course be present. The fidelity witnesses offer an excellent tool for experimentally quantifying, in an an inexpensive way, the detrimental effect of such imperfections on the simulation's performance.

Free-fermionic models are classically tractable, but the importance of their quantum simulations comes from the fact that they constitute a testbed for experimental many-body quantum technologies, with certified simulations of classically intractable models as ultimate goal. In this respect, the direct-certification tools developed here may help bridge the gap between the experimental certification of proof-of-principle simulations and classically intractable ones.

\emph{Acknowledgements}.  
We thank C.\ Krumnow, D.\ Hangleiter, D.\ Gross, and Z.\ Zimboras for fruitful discussions.
The work of JE and MG was funded by the Templeton Foundation, 
the EU (AQuS), the ERC (TAQ), and the DFG (SPP 1798 CoSIP, 
EI 519/7-1, EI 519/9-1, EI 519/14-1 and
GRO 4334/2-1).
The work of MK was funded by the National Science Centre, Poland (Polonez 2015/19/P/ST2/03001) within the European Union's Horizon 2020 research and innovation programme under the Marie Skłodowska-Curie grant agreement No 665778. 
LA acknowledges financial support from the Brazilian agencies CNPq, CAPES, FAPERJ, and FAPESP.
\bibliographystyle{apsrev4-1}


%

\section*{Appendix}

In this appendix, 
we present the technicalities of the calculations mentioned in the main text and additionally provide some further details about our methods.
The first section recalls generally known facts concerning fermionic linear optics.
The next three sections concern fidelity witnesses.
The next two are on sample complexities for evaluating the Gaussian fidelity witness with an estimate.
In the last section we provide details to robustness properties of the fidelity witness and the corresponding certification test.

\subsection{Methods of fermionic linear optics}
\label{app:FLO}
This section gives more details on results of fermionic linear optics used in the main text.
The first sub-subsection discusses unitary evolution in this formalism.
The second sub-subsection contains details on the Jordan-Wigner transformation, covariance matrices of spin product states and which spin operators need to be measured to measure the fermionic covariance matrix.
Finally we shortly comment on the numerical simulations.
\paragraph{Gaussian dynamics}
The Heisenberg evolution of Majorana operators is given as follows.
\begin{lemma}[Free fermion propagator]
Let
\begin{align}
\label{eq:H(A)}
  H(\boldsymbol A) = \tfrac\i4 \sum_{j,k=1}^{2L}A_{j,k} m_jm_k\ 
\end{align}
with $\boldsymbol A=-\boldsymbol A^\t \in \mathbb R^{2L\times 2L}$. Then
\begin{align}
\label{eq:m(t)}
  m_j(t) \equiv e^{\i t H(\boldsymbol A)}m_j e^{-\i t H(\boldsymbol A)} = \sum_{k=1}^{2L}Q_{j,k}(t) m_k
\end{align}
where $\boldsymbol Q(t)=e^{t\boldsymbol A}\in SO(2L)$.
\end{lemma}
Note that the propagator is manifestly real and there is no $\i$ in the exponent because $\boldsymbol A$ is antisymmetric.
\begin{proof}
We begin by noticing that $m_j(t)$ is differentiable and take a time-derivative obtaining
\begin{align}
  \partial_t m_j(t) 
  &= \i H(\boldsymbol A) m_j(t) - m_j(t) H(\boldsymbol A)\\
  &= \i [\ H(\boldsymbol A),\ m_j(t)\ ]
\end{align}
which is the Heisenberg equation of motion.
We further notice that 
\begin{align}
  \partial_t m_j(t) &=  \i e^{\i t H(\boldsymbol{A})} [\ H(\boldsymbol{A}),\ m_j\ ] e^{-\i t H(\boldsymbol{A})}
\end{align}
which means that we need to evaluate the commutator at $t=0$.
Next we calculate the commutator
\begin{align}
  [\ m_{j'}m_k,\ m_j\ ]=2m_{j'}\delta_{k,j}-2m_k\delta_{j',j}
\end{align}
which gives
\begin{align}
  [\ H(\boldsymbol A),\ m_j\ ] &= \tfrac\i4 \sum_{j',k=1}^{2L}A_{j',k} [\ m_{j'}m_k,\ m_j\ ]\\
  &= \tfrac\i2 \sum_{j',k=1}^{2L} ( A_{j',k} m_{j'}\delta_{k,j}-A_{j',k}m_{k} \delta_{j',j}) \\
  &= \tfrac\i2 \sum_{k=1}^{2L} ( A_{k,j} m_{k}-A_{j,k}m_{k}) \\
  &= - \i \sum_{k=1}^{2L}  A_{j,k}m_{k} \ .
\end{align}
This allows us to write the above Heisenberg equation of motion explicitly as
\begin{align}
   \partial_t m_j(t) = \sum_{k=1}^{2L}  A_{j,k}m_{k}\ .
\end{align}
This linear system of $2L$ ordinary differential equations is solved by
\begin{align}
  m_j(t) = \sum_{k=1}^{2L}Q_{j,k}(t) m_k\ ,
\end{align}
where $\boldsymbol Q=e^{t\boldsymbol A}\in SO(2L)$.
Indeed, this becomes apparent if one considers a vector $m = (m_1,\ldots,m_{2L})^\t$ then we get in vector notation
\begin{align}
  \partial_t\ m(t) = \boldsymbol A\ m(t)\quad \Leftrightarrow\quad  m(t) = e^{t\boldsymbol A}\ m\ .
\end{align}
\end{proof}

Given this we easily obtain the evolution equation for the covariance matrix $M(\varrho(t))_{j,k} = \frac{\i}{2}\tr\bigl([m_j,m_k]\,\varrho(t)\bigr) =  \frac{\i}{2}\tr\bigl([m_j(t),m_k(t)]\,\varrho\bigr) = \sum_{j',k'=1}^{2L} Q_{j,j'}(t)Q_{k,k'}(t)M(\varrho(0))_{j',k'}$.
This in matrix notation gives $\boldsymbol M(\varrho(t)) = (\boldsymbol Q(t) \boldsymbol M(\varrho(0))\boldsymbol Q(t)^\t)_{j,k}$.
\paragraph{Using the Jordan-Wigner transformation}
This paragraph shows how to use the Jordan-Wigner transformation to translate between spins and fermions.
We first identify covariance matrices of simple states.
\begin{lemma}[Vacuum covariance matrix]
In the notation $\sigma^z\ket \uparrow = \ket\uparrow$ we have
\begin{align}
	\boldsymbol{\cm}\big(\ketbra{\uparrow}{\uparrow}^{\otimes L}\big)\coloneqq\oplus_{j=1}^L\begin{psmallmatrix} 
  0 & -1 \\ 1 & 0 
  \end{psmallmatrix}.
\end{align}
In general if $\ket {\boldsymbol \omega}$ is a computational basis state with $\boldsymbol \omega\in\{0,1\}^{\times L}$ (identifying $\ket 0 = \ket \uparrow$ and $\ket 1 =\ket \downarrow$)
we have
\begin{align}
	\boldsymbol{\cm}\big(\ketbra{\boldsymbol\omega}{\boldsymbol \omega}\big)\coloneqq\oplus_{j=1}^L\begin{psmallmatrix} 
  0 & -(-1)^{\omega_k} \\ (-1)^{\omega_k} & 0 
  \end{psmallmatrix}.
\end{align}

\end{lemma}

\begin{proof}
The first statement follows directly from the second for $\omega_k=0$ for all $k$.

To show the latter, we first observe that $\sigma^z_k = -\i m_{2k-1}m_{2k}$.
Indeed using 
\begin{align}
\label{eq:Pauli_rules}
  \sigma^a\sigma^b = \delta_{a,b}\id_2 +\i \sum_{c=x,y,z}\varepsilon_{a,b,c} \sigma^c\ 
\end{align}
we get
\begin{align}
-\i m_{2k-1}m_{2k}
  &= - \i \bigl(\prod_{k'<k}\sigma_{k'}^z\bigr) \sigma^x_k \bigl(\prod_{k''<k}\sigma_{k''}^z\bigr) \sigma^y_k\\
  &= -\i \sigma^x_k\sigma^y_k= \sigma^z_k\ .
\end{align}
Next we observe that $\langle \sigma^x_k \rangle_{\boldsymbol{\omega}} = \langle \sigma^y_k\rangle_{\boldsymbol{\omega}} = 0$ and $\langle \sigma^z_k \rangle_{\boldsymbol{\omega}}=(-1)^{\omega_k}$ so the only non-vanishing elements are
\begin{align}
 M_{2k-1,2k} = - M_{2k,2k-1} &= \i \langle m_{2k-1}m_{2k}\rangle_{\boldsymbol{\omega}}\\
 &= -  \langle \sigma^z_k \rangle_{\boldsymbol{\omega}} = -(-1)^{\omega_k}\ .
\end{align}

\end{proof}
In an experiment based on qubits the fermionic covariance matrix can be measured by making the following Pauli measurements.
\begin{lemma}[Fermion spin correlation dictionary]
For $j<k$ we have
\begin{itemize}
\item Odd-odd
\begin{align}
  m_{2j-1}m_{2k-1}=-\i \sigma^y_j\bigl(\prod_{j<k'<k}\sigma_{k'}^z\bigr) \sigma^x_k
\end{align}
\item Odd-even
\begin{align}
  m_{2j-1}m_{2k} =  -\i \sigma^y_j\bigl(\prod_{j<k'<k}\sigma_{k'}^z\bigr) \sigma^y_k
\end{align}
\item Even-odd
\begin{align}
  m_{2j}m_{2k-1} =  \i \sigma^x_j\bigl(\prod_{j<k'<k}\sigma_{k'}^z\bigr) \sigma^x_k
\end{align}

\item Even-even
\begin{align}
  m_{2j}m_{2k} =  \i \sigma^x_j\bigl(\prod_{j<k'<k}\sigma_{k'}^z\bigr) \sigma^y_k\ .
\end{align}
\end{itemize}
\end{lemma}

\begin{proof}

\begin{align}
  m_{2j-1}m_{2k-1} &= \bigl(\prod_{j'<j}\sigma_{j'}^z\bigr) \sigma^x_j \bigl(\prod_{k'<k}\sigma_{k'}^z\bigr) \sigma^x_k\\
  &= \sigma^x_j \sigma^z_j\bigl(\prod_{j<k'<k}\sigma_{k'}^z\bigr) \sigma^x_k\\
  &= -\i \sigma^y_j\bigl(\prod_{j<k'<k}\sigma_{k'}^z\bigr) \sigma^x_k\;.
\end{align}
The remaining relations follow similarly and by again using \eqref{eq:Pauli_rules}.
\end{proof}
Considering the reversed direction of this dictionary, we find that the product of two spin operators is a product of again two Majorana operators only when the spins are neighboring in the Jordan-Wigner transformation from which we obtain the following corollary.
\begin{corollary}[XY models]
The Hamiltonian $H_\mathrm{spin}$ from main text maps to a quadratic fermionic Hamiltonian $H(\boldsymbol A)$ under the Jordan-Wigner transformation.
\end{corollary}
The translation invariant case is  physically the most relevant case for which  the following  result first appeared in \cite{Pfeuty70} and we state it to make explicit which couplings we have used in our simulations.
\begin{lemma}[Transverse field Ising model]
The Hamiltonian of the transverse field Ising model
\begin{equation}
H_\mathrm{TFIM}
=
-J\,\sum_{k=1}^{L-1} \X_k\,\X_{k+1}- B\,\sum_{k=1}^L  \Z_k\, 
\end{equation}
maps to free fermions under the Jordan-Wigner transformation and the couplings matrix read
\begin{align}
\boldsymbol{A}(J,B) = 2\left( \begin{array}{cccccc}
0&B&&&&\\
-B&0&J&&&\\
&-J&0&B&&\\
&&-B&0&J&\\
&&&-J&0&\\
&&&&&\ddots \end{array} \right) .
\end{align}
\end{lemma}
Note, that for compactness we write in the main text $\boldsymbol A(J) \equiv \boldsymbol A (J,0)$ and $\boldsymbol A(B) \equiv \boldsymbol A (0,B)$.
\begin{proof}
By the above dictionary lemma we have $\sigma^x_{k}\sigma^x_{k+1}=-\i m_{2k}m_{2k+1}$ and $\sigma^z_k= -\i m_{2k-1}m_{2k}$.
This gives
$H_\mathrm{TFIM} =  \i \sum_{k=1}^{L-1}J  m_{2k}m_{2k+1} + \i \sum_{k=1}^L B  m_{2k-1}m_{2k}$
which can be put to the standard form $H(\boldsymbol A) = \tfrac\i 4 \sum_{j,k}^{2L} A_{j,k} m_j m_k$ by defining the matrix $\boldsymbol A$ as in the lemma statement.
\end{proof}
\paragraph{Comments on numerics}
The numerical code used to obtain Fig. 2 in main text is available at \cite{github}.
We use Wick's formula \cite{bravyi2004lagrangian} to calculate $|\langle \prod_{k=1}^n\sigma^z_k\rangle| = \mathrm{Pf}( \boldsymbol M_{1\ldots 2n})$ where $\mathrm{Pf}$ denotes the Pffafian which can be calculated using the package PFAPACK \cite{PFAPACK}. 

\subsection{Proof that Eq.\ \eqref{eq:def_W} yields a fidelity witness and general witness construction}
\label{sec:Proof_general_wti_construction}
Here we first provide an expression for a fidelity witness of any arbitrary, totally generic pure target state, not restricted to the Gaussian fermionic setting. 
\begin{proposition}[General witness construction]
\label{prop:witness_construction}
Let $\varrho_\text{t}$ be any pure target state, $0<\Delta=\lambda_1\le\ldots\le \lambda_N$, and $P_1$, $P_2$, $\hdots$, and $P_N$ positive-semidefinite operators such that $\varrho_\text{t}+ \sum_{l=1}^N  P_l=\id$ and  $\tr(\varrho_\text{t}\,  P_l)=0$ for all $l=1, \ldots , N$. 
Then,
\begin{align}
\W\coloneqq\id -\Delta^{-1}\sum_{l=1}^N \lambda_l \, P_l
\end{align}
is a fidelity witness for $\varrho_\text{t}$.
\end{proposition}


The fact that the observable $\W$ in Eq.\ \eqref{eq:def_W} defines a fidelity witness for the free-fermionic target state in Eq.\ \eqref{eq:taregt_state_def} follows from Proposition~\ref{prop:witness_construction} taking $N=2^L-1$, identifying $l$ with an $L$-bit string $\boldsymbol{\nu}\neq\boldsymbol{\omega}$, and taking
$\lambda_{\boldsymbol{\nu}}=\sum_{j=1}^L\left[ (1-\omega_j)  \nu_j +\omega_j(1- \nu_j)\right]$ and $P_{\boldsymbol{\nu}}=U\ket{\boldsymbol{\nu}}\bra{\boldsymbol{\nu}}U^\dagger$. 

\begin{proof}[Proof of Proposition \ref{prop:witness_construction}]
We start with Property \ref{item:iff} in Def.\ \ref{def:witness}.
Let $\varrho_p$ be such that $\tr[\W \varrho_p]=1$.
Then $\Delta^{-1}\sum_{k=1}^N \lambda_l \tr[ P_k\varrho_p]=0$.
As all terms are non-negative, we have $\tr[ P_k\varrho_p]=0$.
From this we write $1=\tr[\varrho_p \id]=\tr[\varrho_p \varrho_\text{t}]+\sum_{k=1}^N \tr[ P_k\varrho_p]=\tr[\varrho_p \varrho_\text{t}]$, which means, since $\varrho_\text{t}$ is pure, that $\varrho_p=\varrho_\text{t}$.
The converse direction starting from $\varrho_p=\varrho_\text{t}$ follows from $\tr[\varrho_\text{t}  P_k]=0$.

We now prove Property \ref{item:geq} in Def.\ \ref{def:witness}.
For any state vector $\ket \psi$ we have
\begin{align}
  \sum_{k=1}^N \lambda_l \langle \psi| {P}_k| \psi \rangle
    &\geq 
  \Delta \sum_{k=1}^N  \langle \psi| {P}_k| \psi \rangle
  \\
    &= \Delta (1-\langle \psi|\varrho_\text{t}| \psi \rangle)
  \, .
\end{align}
This means that 
\begin{align}
  \langle \psi| \varrho_\text{t}| \psi \rangle\geq  \langle \psi| \W| \psi \rangle
\end{align}
which one may write
$\varrho_\text{t}\succeq \W=\id -\Delta^{-1} \sum_{k=1}^N \lambda_l  P_l$, where $\succeq$ denotes semidefinite dominance.
This relation can be used in order to lower bound the fidelity.
If we write the preparation state in its eigenbasis $\varrho_{p}=\sum_k p_k |k\rangle\langle k|\succeq 0$, then we find the following
\begin{align}
\tr[(\varrho_\text{t} -\W)\varrho_{p}]=\sum_k p_k \langle k|\varrho_\text{t} -\W|k\rangle\ge 0\;.
\end{align}
Thus we arrive at
\begin{align}
F=\tr[\varrho_\text{t}\, \varrho_{p}]\ge\tr[\W\, \varrho_{p}]\;.
\end{align}
\end{proof}

\subsection{Proof of Eq.\ \texorpdfstring{\eqref{eq:def_Fgeq}}{(9)}: Fidelity-witness in terms of covariance matrices}
\label{sec:evaluation}
Before the proof, let us first provide useful facts from fermionic linear optics theory.
The covariance matrix of any Fock state vector $\ket{\boldsymbol{\omega}}$ is given, introducing the short-hand notation $\boldsymbol{M}_{\boldsymbol{\omega}}\coloneqq\boldsymbol{M}(\ketbra{\boldsymbol{\omega}}{\boldsymbol{\omega}})$ by
\begin{align}
\label{eq:def_Jw}
  \boldsymbol{M}_{\boldsymbol{\omega}}=\bigoplus_{k=1}^L (1-2w_k)
  \begin{pmatrix} 
  0 & -1 \\ 1 & 0 
  \end{pmatrix} \;. 
\end{align}
This is readily seen from the fact that $\i [m_{2k-1},m_{2k}]/2=(f_k+f_k^\dagger)(f_k-f_k^\dagger)=2  n_k-\id$ which gives $M_{2k-1,2k}=\i\bra{\boldsymbol \omega}  m_{2k-1}m_{2k}\ket{\boldsymbol \omega}=2w_k-1=-M_{2k,2k-1}$ and that all other covariance matrix entries are zero.
Put differently, fermionic Fock states are of the most simple product form.
In order to introduce coherences in the system one can rotate the state by a Gaussian unitary $U$ with mode action $\boldsymbol{Q}$ which then yields
\begin{align}
   \boldsymbol{M}(U\ketbra{\boldsymbol{\omega}}{\boldsymbol{\omega}}U^\dagger)= \boldsymbol{Q} \, \boldsymbol{M}_{\boldsymbol{\omega}}\, \boldsymbol{Q}^\t  .
\end{align}

\begin{proof}[Proof of Eq.\ \eqref{eq:def_Fgeq}]
  In order to evaluate the witness we notice that the numbering operator of mode $k$ is 
   \begin{align}
  n_k=\frac{\id}{2}+\frac{\i}{4}[m_{2k-1},m_{2k}]
   \end{align}
    and    \begin{align}\id -n_k=\frac{\id}{2}-\frac{\i}{4}[m_{2k-1},m_{2k}].   \end{align} 
 This allows us to write the projector $n^{(\boldsymbol{\omega})}$ as 
 \begin{align}
 n^{(\boldsymbol{\omega})}=\sum_{j=1}^L\left[\id/2+ \frac{\i}{4} (1-2\omega_j) [m_{2k-1},m_{2k}])\right ]
\,.
 \end{align}
  We therefore have 
\begin{align}
  \tr(\varrho_p \W)   &=1-\frac{L}{2}
   - \frac{\i}{4}
  \sum_{k=1}^L (1-2w_k)\tr\bigl(U^\dagger\varrho_pU\bigl[ m_{2k-1} , m_{2k} \bigr] \bigr) 
     \\
     &=
     1-\frac{L}{2} - \frac{1}{2}
     \sum_{k=1}^L(1-2w_k)M(U^\dagger\, \varrho_p\, U)_{2k-1,2k} \, ,
\end{align}
where the definition of the covariance matrix \eqref{eq:def_cm} has been used. 
As $\tilde{\boldsymbol{M}} \coloneqq \boldsymbol{M}(U^\dagger\, \varrho_p\, U)=\boldsymbol{Q}^\t \, \boldsymbol{M}(\varrho_\text{p})\, \boldsymbol{Q}$ is anti-symmetric, we can write 
$\tilde\cm_{2k-1,2k}$ as
\begin{equation}
  \tilde \cm _{2k-1,2k} 
  =
  \frac12\tr\Bigl[ 
  \begin{pmatrix}
    0&\tilde \cm _{2k-1,2k}\\
    \tilde \cm _{2k,2k-1} &0
  \end{pmatrix}
  \begin{pmatrix} 
  0 & -1 \\ 1 & 0 
  \end{pmatrix}
  \Bigr].
\end{equation}
We further notice that\ \eqref{eq:def_Jw} allows us to write
\begin{align}
\sum_{k=1}^L(1-2w_k)\tilde M_{2k-1,2k}=\frac12 \tr[\tilde{\boldsymbol{M}}\boldsymbol{M}_{\boldsymbol{\omega}}]\,.
\end{align}
From the definition of $\boldsymbol{M}(\varrho_\text{t})=\boldsymbol{Q} \, \boldsymbol{M}_{\boldsymbol{\omega}}\, \boldsymbol{Q}^\t$ we finally obtain 
\begin{align}
\tr(\varrho_p \W) 
&=
1-\frac{L}{2} - \frac14 \tr[\boldsymbol{M}(\varrho_\text{p})\,\boldsymbol{M}(\varrho_\text{t})]\\
&=1+\frac14 \tr\bigl[(\boldsymbol{M}(\varrho_\text{p})-\boldsymbol{M}(\varrho_\text{t}))^\t \boldsymbol{M}(\varrho_\text{t}) \bigr]\,.
\end{align}

\end{proof}

\subsection{Proof of Theorem~3 (Sample complexity of \texorpdfstring{$\F$}{F})}
\label{sec:single_shot_imp_sampling}

In this section, we compute the number of experimental runs required to get a finite-sample estimate $\F^*(\varrho_\text{p})$ of $\F(\varrho_\text{p})$ satisfying Eq.\ \eqref{eq:large_dev_bound} 
with the measurement scheme with single-shot importance sampling described in the main text. This sets the upper bound on $\mathcal N_{\epsilon,\delta}(\mathcal W)$ in Eq.\ \eqref{eq:unified_exp_sample_complex}, proving Theorem \ref{thm:main}.
\begin{proof}[Proof of Theorem~\ref{thm:main}]
We begin by noting that one can directly evaluate $\F(\varrho_\text{p})$ from the value of 
\begin{align}
\mathcal X\coloneqq\tr\bigl[\boldsymbol{M}(\varrho_\text{p})^\t\,\boldsymbol{M}(\varrho_\text{t})\bigr]=4(F_\W+\frac L 2 -1).
\end{align}
Indeed, if $|\mathcal X^*-\mathcal X|\le 4\epsilon$, then $|\F^*(\varrho_\text{p})-\F(\varrho_\text{p})|\le \epsilon$.
We define conditional probability
\begin{align}
P_{\beta\vert j,k}\coloneqq \tr\Big[\hat  m^{(\beta)}_{j,k}\,\varrho_{p}\Big]
\end{align}
and  the sampling distribution 
\begin{align}
P_{j,k}=\frac{|\cm_{j,k}(\varrho_\text{t})|}{|\boldsymbol{M}(\varrho_\text{t})|}
\end{align}
for $(j,k) \in \Omega$ with $|\boldsymbol{M}(\varrho_\text{t})| = \sum_{(j,k) \in \Omega} |\cm_{j,k}(\varrho_\text{t})|\le 2L^2$.
  By Bayes' theorem, we have that $P_{\beta,j,k}=P_{\beta\vert j,k} P_{j,k}$ is a well-defined probability distribution.
Additionally, we define the importance sampling variable 
\begin{align}
X_{\beta, j,k}\coloneqq 2\, |\boldsymbol{M}(\varrho_\text{t})|\,\beta\,\text{sgn}\big[\cm_{j,k}(\varrho_\text{t})\big]
\end{align}
which is distributed over $P_{\beta,j,k}$.
With these definitions we check that the average of $X$ gives $\mathcal X$
\begin{align}
\label{eq:witness_non_const1}
\mathbb E[ X]
&=\sum_{(j,k)\in \Omega,\beta=\pm1} X_{\beta, j,k}P_{\beta,j,k}\\
  &= 2\sum_{(j,k)\in\Omega} \text{sgn}\big[\cm_{j,k}(\varrho_\text{t})\big]
 |\cm_{j,k}(\varrho_\text{t})|\sum_{\beta=\pm1} \beta \tr[ \hat  m^{(\beta)}_{j,k}\,\varrho^{(p)}]   
\label{eq:bernstein_eval}\\
&=\tr\bigl[\boldsymbol{M}(\varrho_\text{p})^\t\,\boldsymbol{M}(\varrho_\text{t})\bigr]\;.
\end{align} 
We now use Hoeffding's inequality to see that this results in a $(\epsilon,\delta)$-evaluation promise.
We have 
\begin{align}
  \mathbb{P} \left [|\mathcal X -\frac{1}{\mathcal{N}}\sum_{m=1}^\mathcal{N} X_{\mu(m)}| > 4\epsilon\right] \leq 2\exp\left({-\frac{2\,\mathcal{N}\,\epsilon^2}  {|\boldsymbol{M}(\varrho_\text{t})|^2}}\right).
\end{align}
We would like the RHS to be upper bounded by $\delta$ so we obtain
\begin{align}
\mathcal N_{\epsilon,\delta}(\W)=\left\lceil\frac{\ln(2/\delta) |\boldsymbol{M}(\varrho_\text{t})| ^2}{2\,\epsilon^2}\right\rceil
  \label{eq:final_complexity}
  \end{align}
  which is the sample complexity, i.e., yielding the inequality~\eqref{eq:unified_exp_sample_complex}.

\end{proof}

\subsection{Sample complexity for entrywise evaluation}
\label{sec:entrywise_eval}
Here, we compute the number of experimental runs required to get a finite-sample estimate $\F^*(\varrho_\text{p})$ of $\F(\varrho_\text{p})$ satisfying Eq.\ \eqref{eq:large_dev_bound} with a measurement scheme that does not exploit importance sampling, i.e., where all $|\Omega|$ observables are deterministically measured, but that exploits the fact that commuting observables with indices in $\Omega$ can be measured simultaneously in each run. As we show, the resulting bound is less tight than the one in Eq.\ \eqref{eq:unified_exp_sample_complex}.
More precisely, we consider a procedure where all $|\Omega|$ observables are measured the same number of times
 \begin{align}
  \eta=\epsilon^{-2}L^3 \ln(2 |\Omega|/\delta)
 \end{align}
 and we obtain the sample complexity $\mathcal N_{\epsilon,\delta}(\W) =4L \eta$.

 We denote the estimator of $\boldsymbol{\cm}$ by $\boldsymbol{\cm}^\ast$.
The fact that the covariance matrix entries are bounded and lie in the interval $-1< M_\mu<1$ allows us to use Hoeffding's inequality.
Taking  $b= 
 \ln(2|\Omega|/\delta)
$ and making a union bound we find
\begin{align}
\label{eq:exponential_bound}
\mathbb P\Bigl[ 
\forall \mu \in \Omega: \bigl| {\cm}_\mu -{\cm}_\mu^\ast \bigr| 
\leq
\sqrt{2b/\eta}
\Bigr]&=\\
1-\mathbb P\Bigl[ 
\exists \mu \in \Omega: \bigl| {\cm}_\mu -{\cm}_\mu^\ast\bigr| 
\geq
\sqrt{2b/\eta}
\Bigr]&\ge\\
1-|\Omega|\max_{\mu\in\Omega}\mathbb P\Bigl[ 
 \bigl| {\cm}_\mu -{\cm}_\mu^\ast \bigr| 
\geq
\sqrt{2b/\eta}
\Bigr]
&\geq 1-2|\Omega|\e^{-b},
\end{align}
where we have used that $\mathbb{P}[A\cup B]\le \mathbb P [A]+\mathbb P[B]$ for any probability measure $\mathbb P$.
We check that $2|\Omega|\e^{-b}=\delta$
and additionally
\begin{align}
  \sqrt{2 b/\eta}=\sqrt{2\epsilon^2L^{-3}}=\sqrt 2 L^{-3/2}\epsilon
\end{align}
and therefore we have \begin{align}
\mathbb P\left[ 2^{-1/2}L^{3/2}\mnorm{\boldsymbol{\cm} -\boldsymbol{\cm}^\ast} \le \epsilon \right]\ge 1-\delta.
\end{align}
Eq.\ \eqref{eq:large_dev_bound} follows thanks to the following Lemma which tells us that one can efficiently estimate the fidelity lower bound from estimates of the covariance matrix of $\varrho_p$ with small errors.
 
\begin{lemma}[Stability]\label{prop:stability}
The fidelity lower bound $\F(\varrho_\text{p})$ is Lipschitz continuous with Lipschitz constant $L^{3/2}/\sqrt{2}$ with respect to the max-norm, i.e.,  
for any two covariance matrices $\boldsymbol{\cm} $ and $\boldsymbol{\cm}^*$ we have for the respective values of the fidelity witnesses 
\begin{align}
\label{eq:stability}
|\F(\varrho_\text{p})-\F^*(\varrho_\text{p})|
\leq
2^{-1/2} L^{3/2} \mnorm{\boldsymbol{\cm} -\boldsymbol{\cm}^\ast} \, .
\end{align}
\end{lemma}
In the following proof, we denote the trace-norm by $\norm{\argdot}_1$,
the Schatten $2$-norm (or Frobenius norm) by $\norm{\argdot}_2$, and 
the spectral norm by $\snorm{\argdot}$.

\begin{proof}[Proof of Lemma~\ref{prop:stability}]
  Let 
   \begin{align}
  \boldsymbol{\J}_L=\oplus_{k=1}^L  \begin{pmatrix} 
  0 & 1 \\ -1 & 0 
  \end{pmatrix}.
    \end{align} 
It is enough to show that the linear map $\boldsymbol{\cm}  \mapsto \tr[\boldsymbol{Q}\,\boldsymbol{\cm}\, \boldsymbol{Q}^\t\, \boldsymbol{\J}_L]$  is Lipschitz continuous at the origin with Lipschitz constant $(2L)^{3/2}$. 
  
  By H\"olders inequality we have
  \begin{align}
    \bigl|\tr[\boldsymbol{Q}\,\boldsymbol{\cm}\, \boldsymbol{Q}^\t\, \boldsymbol{\J}_L]\bigr|
    &= 
    \bigl|\tr[\boldsymbol{\cm}\, \boldsymbol{Q}^\t\, \boldsymbol{\J}_L \,\boldsymbol{Q}]\bigr|
    \\
    &\leq 
    \norm{\boldsymbol{\cm}}_1 \snorm{\boldsymbol{Q}^\t\, \boldsymbol{\J}_L \,\boldsymbol{Q}} =  \norm{\boldsymbol{\cm}}_1 \, ,
    \label{eq:bound_M_tr_norm}
  \end{align}
where we have used that 
$\norm{\argdot}_\infty$ unitarily invariant and that $\norm{\boldsymbol{\J}_L}_\infty =1$
in the last step.
It remains to show that $\norm{\boldsymbol{\cm}}_1\leq 2L$. 
But for any $2L \times 2L$ matrix $\boldsymbol{\cm}$ it holds that 
\begin{align}
  \norm{\boldsymbol{\cm}}_1 \leq \sqrt{2L} \norm{\boldsymbol{\cm}}_2 \leq \sqrt{2L}\, 2L \mnorm{\boldsymbol{\cm}},
  \label{eq:bound_M_max_norm}
\end{align}
where we have used 
(i) a general norm inequality for the Schatten $1$- and $2$-norm, 
(ii) that the Schatten $2$-norm is the same as the vector $2$-norm of the vectorized matrix, 
(iii) a general norm inequality for the vector $2$-norm and the vector $\infty$-norm, and
(iv) that the vector $\infty$-norm of a vectorized matrix is the max-norm of the matrix. 
Note that the bound \eqref{eq:bound_M_max_norm} is tight for general matrices, as can be seen by choosing $\boldsymbol{\cm}$ as the discrete Fourier transform matrix on $\CC^{2L}$. 
Inserting Eq.~\eqref{eq:bound_M_max_norm} into \eqref{eq:bound_M_tr_norm} completes the proof. 
\end{proof}

Finally, in order to derive the sample complexity, we need to partition the set $[2L]\times[2L]$ such that the corresponding elements of the covariance matrix commute.
We do it by considering bands parallel to the diagonal of the covariance matrix.
Let us consider the non-trivial elements closest to the diagonal $\mu=(i,i+1)$. 
We bi-partition this band into indices starting with an even or an odd number.
By construction, all associated covariance matrix observables will commute.
As there are in total $2L-1$ such off-diagonals, the total number of i.i.d. state preparations is bounded by
\begin{align}
\label{eq:l_variance}
  \mathcal N_{\epsilon,\delta}(\W) =4L \eta=\frac{4 L^4 \ln(2 |\Omega|/\delta)}{\epsilon^2}\;.
\end{align}
Since $|\Omega|\le2L^2+L$, this scaling is logarithmically worse in $L$ than in Eq.\ \eqref{eq:unified_exp_sample_complex}.

\subsection{Robustness of the certification test}
\label{sec:robustness}

Ref.~\cite{aolita2015reliable} established a framework of certification where the notion of \emph{robust quantum state certification} has been defined.
In particular, in such a certification test, one desires to accept states above a threshold fidelity $F_{\th}$ and requires to reject states below $F_{\th}$. 
But a realistic certification test cannot resolve fidelities $F$ very close to $F_{\th}$ and thus one needs to allow for a fidelity region that remains undetermined. 
This idea leads to a robust certification test \cite{aolita2015reliable,Hangleiter16}, where one allows for a fidelity gap $\Delta<1-F_{\th}$. 
A \emph{robust} test is guaranteed to accept a preparation $\varrho_p$ if $F\ge F_{\th}+\Delta$, to reject it if $F< F_{\th}$, and possibly accept it in the intermediate region.
These conditions for the test concern the exact fidelity and need to be translated to a statement concerning the estimate of the witness $\F^*(\varrho_\text{p})$.
\def\S{\mathcal S_\perp(\Delta, \epsilon)}
We will show that for all preparations $\varrho_p$ in a certain class of states $\S$ it suffices to compare the estimator $\F^*(\varrho_\text{p})$ to the number $F_\text{T}+\epsilon$.
In other words, such test is robust 
\begin{enumerate}[i)]
\item \label{item:reject}
if $\varrho_p$ is such that $F<F_{\th}$ then in the same time the witness will testify this i.e. $\F^*(\varrho_\text{p})< F_{\th}+\epsilon$.
This means that whenever the test \emph{has to} reject a preparation, then it will. 
\item \label{item:accept}
if $F\ge F_{\th}+\Delta$ then we have $\F^*(\varrho_\text{p})\ge F_{\th}+\epsilon$.
That is, whenever the fidelity is larger then the threshold fidelity enlarged by the fidelity gap, then the preparation is accepted by the test.
\end{enumerate}
Note, that if $F\in[F_{\th}, F_{\th}+\Delta]$, then the certification test might accept or reject the preparation. 
Specifically, the class $\S$ characterizes the set of preparations $\varrho_p$ where the witness behaves as a weak oracle separating $F\le F_\text{T}$ from $F\ge F_\text{T}+\Delta$. 
We now construct this class.
With a given target  state $U\ket{\boldsymbol{\omega}}$ and its corresponding fidelity witness in mind, we define a mismatch parameter of some preparation state $\varrho_p$ to be
\begin{align}\label{eq:def_mismatch}
n_\perp(\varrho_p)
\coloneqq
\tr[\hat n^{(\boldsymbol{\omega})}\,U^\dagger\, \varrho_p\, U]\ge0\;.
\end{align}
Let us note that the preparation $\varrho_p$ can be decomposed with the Hilbert-Schmidt inner product into the target state $\varrho_\text{t}$ and a orthogonal contribution $\varrho_\perp\coloneqq\varrho_\perp(\varrho_\text{t}, \varrho_p)$, that is $\varrho_p=F \varrho_\text{t} +(1-F)\varrho_\perp$ for $0\le F=\tr[\varrho_p\, \varrho_\text{t}] \le 1$ and $\tr[\varrho_\text{t}\,\varrho_\perp]=0$.
Using linearity of our witness for this decomposition yields
\begin{align}
\F(\varrho_\text{p})&=F+(1-F)(1- \tr[U\,\hat n^{(\boldsymbol{\omega})}\,U^\dagger\, \varrho_\perp])\nonumber\\
&=1-(1-F)n_\perp(\varrho_\perp)\;.\label{eq:FW_mismatch}
\end{align}
Therefore the mismatch content has the properties $n_\perp(\varrho_\text{t})=0$ and $n_\perp(\varrho_p)=(1-F)n_\perp(\varrho_\perp)$.
For a given maximum estimation error
 $0<\epsilon<(1-F_\text{T})/2$, {fidelity gap} $\Delta>2\,\epsilon$ and {fidelity threshold} $F_\text{T}<1$ we define the \emph{mismatch content threshold}
\begin{align}\label{eq:def_n_perp_th}
n_{\perp,\th}(\Delta,\epsilon)\coloneqq\frac{1-F_\text{T}-2\,\epsilon}{1-F_{\th}-\Delta}\;.
\end{align}
 This allows us to consider the following subset of all states $\mathcal S$
\begin{equation}\label{eq:low_mismatch_states}
  \S (\Delta,\epsilon)
  =
  \{ \varrho \in \mathcal S  \,|\, n_\perp(\varrho) \le n_{\perp,\th}(\Delta,\epsilon)\}\;.
\end{equation}
It is a convex set containing mixtures of states with possibly very large mismatch content $n_{\perp,\th}$ and which  includes the target state $\varrho_\text{t}$ in its interior.
The following theorem states that our fidelity witness leads to a robust certification test.

\begin{theorem}[Robust certification of pure Gaussian states]
	Let $F_\text{T}<1$ be a threshold fidelity, $\delta>0$ a maximal failure probability , $0<\epsilon<(1-F_\text{T})/2$ a maximal estimation error and $2\,\epsilon<\Delta<1-F_\text{T}$ a fidelity gap.
	Let  $\varrho_\text{t}$ be a pure Gaussian state and $\varrho_p$ a preparation state.
	Let $\F^*(\varrho_\text{p})$ be the estimator of the fidelity witness from Theorem~\ref{thm:main}. 
	The test accepting the preparation if $\F^*(\varrho_\text{p})\ge F_\text{T}+\epsilon$ and rejecting if $\F^*(\varrho_\text{p}) < F_\text{T}+\epsilon$ yields a  robust certification of $\varrho_\text{t}$ if $\varrho_p\in \S$. 
  For states with high enough fidelity $F>1-L^{-2}$ the witness yields a non-trivial lower bound $\F(\varrho_\text{p})\ge 0$. 
\end{theorem} 

\begin{proof}
The rejection Property \ref{item:reject}) follows by observing that according to Theorem~\ref{thm:main} we have with probability at least $1-\delta$ that 
\begin{align}
  |\F^*(\varrho_\text{p})-\F(\varrho_\text{p})|\le\epsilon 
\end{align}
from which it follows that
\begin{align}
 \F^*(\varrho_\text{p})-\epsilon
 \le
 \F(\varrho_\text{p})\;.
\end{align}
Next we use that the fidelity witness is a lower bound to the fidelity $F$ and that in case \ref{item:reject}) we have $ F  < F_{\th}$ to get the chain
\begin{align}
  \F^*(\varrho_\text{p})-\epsilon \le  \F(\varrho_\text{p}) \le F < F_\text{T}\;.
\end{align}
Therefore, if   $ F  < F_{\th}$ then $\F^*(\varrho_\text{p}) < F_\text{T}+\epsilon$.
In this step we did not need to assume anything on the preparation $\varrho_p$.

Secondly, we show that the test has the acceptance Property \ref{item:accept}) as well.
We now use $F\geq F_\text{T} + \Delta$ and assume that $n_\perp(\varrho_p)\le n_{\perp,\th}$ to obtain from Eq.~\eqref{eq:FW_mismatch}
\begin{align}
\F(\varrho_\text{p}) &=1-(1-F) n_\perp \nonumber \\
&\ge  1-n_{\perp,\th}+Fn_{\perp,\th} \nonumber\\
&\ge  1-n_{\perp,\th}+(F_\text{T}+\Delta)n_{\perp,\th} \, 
\end{align}
 which with the definition of the mismatch content \eqref{eq:def_n_perp_th} 
becomes
\begin{align}
\F(\varrho_\text{p}) &\ge  1-(1-F_\text{T}-\Delta)n_{\perp,\th} \\
&\ge 1-(1- F_{\th}-2\,\epsilon) \\
&\ge  F_{\th}+2\,\epsilon \, .
\end{align}
Therefore, we find with probability at least $1-\delta$ the inequality for the estimator of the fidelity witness
\begin{align}
\F^*(\varrho_\text{p})\ge \F(\varrho_\text{p})-\epsilon\ge F_{\th}+\epsilon\;.
\end{align}
These two steps show that the test is robust for $\varrho\in \S$.

Finally let us assume $F_\text{T}\ge1- L^{-2}$. 
We need to show that if $F\ge F_\text{T}+\Delta$ then $\F^*(\varrho_\text{p})\ge F_\text{T}+\epsilon$.
By Fuchs-van der Graaf inequality we have $D(\varrho_p,\varrho_\text{t})\le \sqrt{1-F}\le\sqrt{1-F_\text{T}}=L^{-1}$.
From this bound it follows that 
\begin{align}
\eta&=\sum_{k=1}^L\left[(1-w_k)\tr[n_kU^\dagger\varrho_pU]+w_k\tr[(\id-n_k)U^\dagger\varrho_pU]\right]\\
&\le\sum_{k=1}^L L^{-1} =1.
\end{align}
From this bound, we find
\begin{align}
\F(\varrho_\text{p}) =1-\eta \ge 0 \;.
\end{align}
\end{proof}

Note that the witness is exact $\F(\varrho_\text{p})=F$ for preparation state vectors supported on the Hilbert space subspace $\text{Span}(\{U\vacket, U f_1^\dagger\vacket,\ldots,U f_L^\dagger\vacket\})$.
Finally, Eq.\ \eqref{eq:FW_mismatch} allows to intuitively understand when exactly the witness fails to be an oracle, which we illustrate with one last example.

\subsection*{ Example: symmetry breaking}
Consider a scenario, where the system has initially a $\ZZ_2$ symmetry between the vacuum $\vacket$ and the fully occupied state 
vector $\ket{\overline 1}$, and then at some point \emph{spontaneous symmetry breaking} 
occurs such that the system must choose one of the two states.
If the preparation is given by $ \varrho_p=U\left[(1-\lambda)\vac + \lambda | \overline 1 \rangle\langle \overline 1|\right]U^\dagger$ then the mismatch is a very good way of quantifying the fidelity of symmetry breaking, namely $n_\perp(\varrho_p)=\lambda\langle \overline 1|\hat N |\overline 1\rangle=\lambda L$ is a good order parameter.
The mismatch is low for $\lambda \ll 1/L$  which occurs for high values of our witness and it therefore allows to show that the system  chose the vacuum in the $\mathbb Z_2$ symmetry breaking.
Note, that the mismatch parameter will be high for many-particle GHZ states, but those are expected to be unstable and will not occur for no reason e.g. due to incoherent noise.
In particular, low mismatch is also a natural assumption when certifying a digital simulation of the transverse field Ising Hamiltonian. 
As a final corollary to this example, note that for an $L$-mode system we have $\| n^{(\boldsymbol{\omega})}\|=L$ and therefore for all states $\varrho$ in a ball defined by $\|\varrho-\varrho_\text{t}\|_1\le 1/L$ we will find $\F(\varrho_\text{p})\ge0$.


\begin{thebibliography}{67}%
\makeatletter
\providecommand \@ifxundefined [1]{%
 \@ifx{#1\undefined}
}%
\providecommand \@ifnum [1]{%
 \ifnum #1\expandafter \@firstoftwo
 \else \expandafter \@secondoftwo
 \fi
}%
\providecommand \@ifx [1]{%
 \ifx #1\expandafter \@firstoftwo
 \else \expandafter \@secondoftwo
 \fi
}%
\providecommand \natexlab [1]{#1}%
\providecommand \enquote  [1]{``#1''}%
\providecommand \bibnamefont  [1]{#1}%
\providecommand \bibfnamefont [1]{#1}%
\providecommand \citenamefont [1]{#1}%
\providecommand \href@noop [0]{\@secondoftwo}%
\providecommand \href [0]{\begingroup \@sanitize@url \@href}%
\providecommand \@href[1]{\@@startlink{#1}\@@href}%
\providecommand \@@href[1]{\endgroup#1\@@endlink}%
\providecommand \@sanitize@url [0]{\catcode `\\12\catcode `\$12\catcode
  `\&12\catcode `\#12\catcode `\^12\catcode `\_12\catcode `\%12\relax}%
\providecommand \@@startlink[1]{}%
\providecommand \@@endlink[0]{}%
\providecommand \url  [0]{\begingroup\@sanitize@url \@url }%
\providecommand \@url [1]{\endgroup\@href {#1}{\urlprefix }}%
\providecommand \urlprefix  [0]{URL }%
\providecommand \Eprint [0]{\href }%
\providecommand \doibase [0]{http://dx.doi.org/}%
\providecommand \selectlanguage [0]{\@gobble}%
\providecommand \bibinfo  [0]{\@secondoftwo}%
\providecommand \bibfield  [0]{\@secondoftwo}%
\providecommand \translation [1]{[#1]}%
\providecommand \BibitemOpen [0]{}%
\providecommand \bibitemStop [0]{}%
\providecommand \bibitemNoStop [0]{.\EOS\space}%
\providecommand \EOS [0]{\spacefactor3000\relax}%
\providecommand \BibitemShut  [1]{\csname bibitem#1\endcsname}%
\let\auto@bib@innerbib\@empty
\bibitem [{\citenamefont {Feynman}(1982)}]{Fey82}%
  \BibitemOpen
  \bibfield  {author} {\bibinfo {author} {\bibfnamefont {R.~P.}\ \bibnamefont
  {Feynman}},\ }\bibinfo {title} {\emph {Simulating physics with computers}},\
  \href@noop {} {\bibfield  {journal} {\bibinfo  {journal} {Int. J. Theor.
  Phys.}\ }\textbf {\bibinfo {volume} {21}},\ \bibinfo {pages} {467} (\bibinfo
  {year} {1982})}\BibitemShut {NoStop}%
\bibitem [{\citenamefont {Cirac}\ and\ \citenamefont
  {Zoller}(2012)}]{cirac2012goals}%
  \BibitemOpen
  \bibfield  {author} {\bibinfo {author} {\bibfnamefont {J.~I.}\ \bibnamefont
  {Cirac}}\ and\ \bibinfo {author} {\bibfnamefont {P.}~\bibnamefont {Zoller}},\
  }\bibinfo {title} {\emph {Goals and opportunities in quantum simulation}},\
  \href@noop {} {\bibfield  {journal} {\bibinfo  {journal} {Nature Phys.}\
  }\textbf {\bibinfo {volume} {8}},\ \bibinfo {pages} {264} (\bibinfo {year}
  {2012})}\BibitemShut {NoStop}%
\bibitem [{\citenamefont {Aspuru-Guzik}\ and\ \citenamefont
  {Walther}(2012)}]{aspuru2012photonic}%
  \BibitemOpen
  \bibfield  {author} {\bibinfo {author} {\bibfnamefont {A.}~\bibnamefont
  {Aspuru-Guzik}}\ and\ \bibinfo {author} {\bibfnamefont {P.}~\bibnamefont
  {Walther}},\ }\bibinfo {title} {\emph {Photonic quantum simulators}},\
  \href@noop {} {\bibfield  {journal} {\bibinfo  {journal} {Nature Phys.}\
  }\textbf {\bibinfo {volume} {8}},\ \bibinfo {pages} {285} (\bibinfo {year}
  {2012})}\BibitemShut {NoStop}%
\bibitem [{\citenamefont {Blatt}\ and\ \citenamefont
  {Roos}(2012{\natexlab{a}})}]{blatt2012quantum}%
  \BibitemOpen
  \bibfield  {author} {\bibinfo {author} {\bibfnamefont {R.}~\bibnamefont
  {Blatt}}\ and\ \bibinfo {author} {\bibfnamefont {C.}~\bibnamefont {Roos}},\
  }\bibinfo {title} {\emph {Quantum simulations with trapped ions}},\
  \href@noop {} {\bibfield  {journal} {\bibinfo  {journal} {Nature Phys.}\
  }\textbf {\bibinfo {volume} {8}},\ \bibinfo {pages} {277} (\bibinfo {year}
  {2012}{\natexlab{a}})}\BibitemShut {NoStop}%
\bibitem [{\citenamefont {Houck}\ \emph {et~al.}(2012)\citenamefont {Houck},
  \citenamefont {T{\"u}reci},\ and\ \citenamefont
  {Koch}}]{houck2012chip_review}%
  \BibitemOpen
  \bibfield  {author} {\bibinfo {author} {\bibfnamefont {A.~A.}\ \bibnamefont
  {Houck}}, \bibinfo {author} {\bibfnamefont {H.~E.}\ \bibnamefont
  {T{\"u}reci}}, \ and\ \bibinfo {author} {\bibfnamefont {J.}~\bibnamefont
  {Koch}},\ }\bibinfo {title} {\emph {On-chip quantum simulation with
  superconducting circuits}},\ \href@noop {} {\bibfield  {journal} {\bibinfo
  {journal} {Nature Phys.}\ }\textbf {\bibinfo {volume} {8}},\ \bibinfo {pages}
  {292} (\bibinfo {year} {2012})}\BibitemShut {NoStop}%
\bibitem [{\citenamefont {Bloch}\ \emph {et~al.}(2012)\citenamefont {Bloch},
  \citenamefont {Dalibard},\ and\ \citenamefont
  {Nascimbene}}]{bloch2012quantum}%
  \BibitemOpen
  \bibfield  {author} {\bibinfo {author} {\bibfnamefont {I.}~\bibnamefont
  {Bloch}}, \bibinfo {author} {\bibfnamefont {J.}~\bibnamefont {Dalibard}}, \
  and\ \bibinfo {author} {\bibfnamefont {S.}~\bibnamefont {Nascimbene}},\
  }\bibinfo {title} {\emph {Quantum simulations with ultracold quantum
  gases}},\ \href@noop {} {\bibfield  {journal} {\bibinfo  {journal} {Nature
  Phys.}\ }\textbf {\bibinfo {volume} {8}},\ \bibinfo {pages} {267} (\bibinfo
  {year} {2012})}\BibitemShut {NoStop}%
\bibitem [{\citenamefont {Schneider}\ \emph
  {et~al.}(2012{\natexlab{a}})\citenamefont {Schneider}, \citenamefont
  {Porras},\ and\ \citenamefont {Schaetz}}]{Schneider12}%
  \BibitemOpen
  \bibfield  {author} {\bibinfo {author} {\bibfnamefont {C.}~\bibnamefont
  {Schneider}}, \bibinfo {author} {\bibfnamefont {D.}~\bibnamefont {Porras}}, \
  and\ \bibinfo {author} {\bibfnamefont {T.}~\bibnamefont {Schaetz}},\
  }\bibinfo {title} {\emph {Experimental quantum simulations of many-body
  physics with trapped ions}},\ \href@noop {} {\bibfield  {journal} {\bibinfo
  {journal} {Rep. Prog. Phys.}\ }\textbf {\bibinfo {volume} {75}},\ \bibinfo
  {pages} {024401} (\bibinfo {year} {2012}{\natexlab{a}})}\BibitemShut
  {NoStop}%
\bibitem [{\citenamefont {{Eisert}}\ \emph {et~al.}(2015)\citenamefont
  {{Eisert}}, \citenamefont {{Friesdorf}},\ and\ \citenamefont
  {{Gogolin}}}]{EisFriGog15}%
  \BibitemOpen
  \bibfield  {author} {\bibinfo {author} {\bibfnamefont {J.}~\bibnamefont
  {{Eisert}}}, \bibinfo {author} {\bibfnamefont {M.}~\bibnamefont
  {{Friesdorf}}}, \ and\ \bibinfo {author} {\bibfnamefont {C.}~\bibnamefont
  {{Gogolin}}},\ }\bibinfo {title} {\emph {Quantum many-body systems out of
  equilibrium}},\ \href@noop {} {\bibfield  {journal} {\bibinfo  {journal}
  {Nature Phys.}\ }\textbf {\bibinfo {volume} {11}},\ \bibinfo {pages} {124}
  (\bibinfo {year} {2015})}\BibitemShut {NoStop}%
\bibitem [{\citenamefont {Friedenauer}\ \emph {et~al.}(2008)\citenamefont
  {Friedenauer}, \citenamefont {Schmitz}, \citenamefont {Glueckert},
  \citenamefont {Porras},\ and\ \citenamefont {Schaetz}}]{Friedenauer08}%
  \BibitemOpen
  \bibfield  {author} {\bibinfo {author} {\bibfnamefont {H.}~\bibnamefont
  {Friedenauer}}, \bibinfo {author} {\bibfnamefont {H.}~\bibnamefont
  {Schmitz}}, \bibinfo {author} {\bibfnamefont {J.}~\bibnamefont {Glueckert}},
  \bibinfo {author} {\bibfnamefont {D.}~\bibnamefont {Porras}}, \ and\ \bibinfo
  {author} {\bibfnamefont {T.}~\bibnamefont {Schaetz}},\ }\bibinfo {title}
  {\emph {Simulating a quantum magnet with trapped ions}},\ \href@noop {}
  {\bibfield  {journal} {\bibinfo  {journal} {Nature Phys.}\ }\textbf {\bibinfo
  {volume} {4}},\ \bibinfo {pages} {757} (\bibinfo {year} {2008})}\BibitemShut
  {NoStop}%
\bibitem [{\citenamefont {Kim}\ \emph {et~al.}(2010)\citenamefont {Kim},
  \citenamefont {Chang}, \citenamefont {Korenblit}, \citenamefont {Islam},
  \citenamefont {Edwards}, \citenamefont {Freericks}, \citenamefont {Lin},
  \citenamefont {Duan},\ and\ \citenamefont {Monroe}}]{Kim10}%
  \BibitemOpen
  \bibfield  {author} {\bibinfo {author} {\bibfnamefont {K.}~\bibnamefont
  {Kim}}, \bibinfo {author} {\bibfnamefont {M.-S.}\ \bibnamefont {Chang}},
  \bibinfo {author} {\bibfnamefont {S.}~\bibnamefont {Korenblit}}, \bibinfo
  {author} {\bibfnamefont {R.}~\bibnamefont {Islam}}, \bibinfo {author}
  {\bibfnamefont {E.~E.}\ \bibnamefont {Edwards}}, \bibinfo {author}
  {\bibfnamefont {J.~K.}\ \bibnamefont {Freericks}}, \bibinfo {author}
  {\bibfnamefont {G.-D.}\ \bibnamefont {Lin}}, \bibinfo {author} {\bibfnamefont
  {L.-M.}\ \bibnamefont {Duan}}, \ and\ \bibinfo {author} {\bibfnamefont
  {C.}~\bibnamefont {Monroe}},\ }\bibinfo {title} {\emph {Quantum simulation of
  frustrated Ising spins with trapped ions}},\ \href@noop {} {\bibfield
  {journal} {\bibinfo  {journal} {Nature}\ }\textbf {\bibinfo {volume} {465}},\
  \bibinfo {pages} {590} (\bibinfo {year} {2010})}\BibitemShut {NoStop}%
\bibitem [{\citenamefont {Islam}\ \emph {et~al.}(2011)\citenamefont {Islam},
  \citenamefont {Edwards}, \citenamefont {Kim}, \citenamefont {Korenblit},
  \citenamefont {Noh}, \citenamefont {Carmichael}, \citenamefont {Lin},
  \citenamefont {Duan}, \citenamefont {Wang}, \citenamefont {Freericks},\ and\
  \citenamefont {Monroe}}]{Islam11}%
  \BibitemOpen
  \bibfield  {author} {\bibinfo {author} {\bibfnamefont {R.}~\bibnamefont
  {Islam}}, \bibinfo {author} {\bibfnamefont {E.~E.}\ \bibnamefont {Edwards}},
  \bibinfo {author} {\bibfnamefont {K.}~\bibnamefont {Kim}}, \bibinfo {author}
  {\bibfnamefont {S.}~\bibnamefont {Korenblit}}, \bibinfo {author}
  {\bibfnamefont {C.}~\bibnamefont {Noh}}, \bibinfo {author} {\bibfnamefont
  {H.}~\bibnamefont {Carmichael}}, \bibinfo {author} {\bibfnamefont {G.-D.}\
  \bibnamefont {Lin}}, \bibinfo {author} {\bibfnamefont {L.-M.}\ \bibnamefont
  {Duan}}, \bibinfo {author} {\bibfnamefont {C.-C.~J.}\ \bibnamefont {Wang}},
  \bibinfo {author} {\bibfnamefont {J.~K.}\ \bibnamefont {Freericks}}, \ and\
  \bibinfo {author} {\bibfnamefont {C.}~\bibnamefont {Monroe}},\ }\bibinfo
  {title} {\emph {Onset of a quantum phase transition with a trapped ion
  quantum simulator}},\ \href@noop {} {\bibfield  {journal} {\bibinfo
  {journal} {Nature Comm.}\ }\textbf {\bibinfo {volume} {2}},\ \bibinfo {pages}
  {377} (\bibinfo {year} {2011})}\BibitemShut {NoStop}%
\bibitem [{\citenamefont {Lanyon}\ \emph {et~al.}(2017)\citenamefont {Lanyon},
  \citenamefont {Maier}, \citenamefont {Holz{\"a}pfel}, \citenamefont
  {Baumgratz}, \citenamefont {Hempel}, \citenamefont {Jurcevic}, \citenamefont
  {Dhand}, \citenamefont {Buyskikh}, \citenamefont {Daley}, \citenamefont
  {Cramer} \emph {et~al.}}]{lanyon2016efficient}%
  \BibitemOpen
  \bibfield  {author} {\bibinfo {author} {\bibfnamefont {B.}~\bibnamefont
  {Lanyon}}, \bibinfo {author} {\bibfnamefont {C.}~\bibnamefont {Maier}},
  \bibinfo {author} {\bibfnamefont {M.}~\bibnamefont {Holz{\"a}pfel}}, \bibinfo
  {author} {\bibfnamefont {T.}~\bibnamefont {Baumgratz}}, \bibinfo {author}
  {\bibfnamefont {C.}~\bibnamefont {Hempel}}, \bibinfo {author} {\bibfnamefont
  {P.}~\bibnamefont {Jurcevic}}, \bibinfo {author} {\bibfnamefont
  {I.}~\bibnamefont {Dhand}}, \bibinfo {author} {\bibfnamefont
  {A.}~\bibnamefont {Buyskikh}}, \bibinfo {author} {\bibfnamefont
  {A.}~\bibnamefont {Daley}}, \bibinfo {author} {\bibfnamefont
  {M.}~\bibnamefont {Cramer}},  \emph {et~al.},\ }\bibinfo {title} {\emph
  {Efficient tomography of a quantum many-body system}},\ \href@noop {}
  {\bibfield  {journal} {\bibinfo  {journal} {Nature Physics}\ }\textbf
  {\bibinfo {volume} {13}},\ \bibinfo {pages} {1158} (\bibinfo {year}
  {2017})}\BibitemShut {NoStop}%
\bibitem [{\citenamefont {Barends}\ \emph {et~al.}(2016)\citenamefont
  {Barends}, \citenamefont {Shabani}, \citenamefont {Lamata}, \citenamefont
  {Kelly}, \citenamefont {Mezzacapo}, \citenamefont {Las~Heras}, \citenamefont
  {Babbush}, \citenamefont {Fowler}, \citenamefont {Campbell}, \citenamefont
  {Chen} \emph {et~al.}}]{barends2016digitized}%
  \BibitemOpen
  \bibfield  {author} {\bibinfo {author} {\bibfnamefont {R.}~\bibnamefont
  {Barends}}, \bibinfo {author} {\bibfnamefont {A.}~\bibnamefont {Shabani}},
  \bibinfo {author} {\bibfnamefont {L.}~\bibnamefont {Lamata}}, \bibinfo
  {author} {\bibfnamefont {J.}~\bibnamefont {Kelly}}, \bibinfo {author}
  {\bibfnamefont {A.}~\bibnamefont {Mezzacapo}}, \bibinfo {author}
  {\bibfnamefont {U.}~\bibnamefont {Las~Heras}}, \bibinfo {author}
  {\bibfnamefont {R.}~\bibnamefont {Babbush}}, \bibinfo {author} {\bibfnamefont
  {A.}~\bibnamefont {Fowler}}, \bibinfo {author} {\bibfnamefont
  {B.}~\bibnamefont {Campbell}}, \bibinfo {author} {\bibfnamefont
  {Y.}~\bibnamefont {Chen}},  \emph {et~al.},\ }\bibinfo {title} {\emph
  {Digitized adiabatic quantum computing with a superconducting circuit}},\
  \href@noop {} {\bibfield  {journal} {\bibinfo  {journal} {Nature}\ }\textbf
  {\bibinfo {volume} {534}},\ \bibinfo {pages} {222} (\bibinfo {year}
  {2016})}\BibitemShut {NoStop}%
\bibitem [{\citenamefont {Salath\'e}\ \emph {et~al.}(2015)\citenamefont
  {Salath\'e}, \citenamefont {Mondal}, \citenamefont {Oppliger}, \citenamefont
  {Heinsoo}, \citenamefont {Kurpiers}, \citenamefont
  {Poto\ifmmode~\check{c}\else \v{c}\fi{}nik}, \citenamefont {Mezzacapo},
  \citenamefont {Las~Heras}, \citenamefont {Lamata}, \citenamefont {Solano},
  \citenamefont {Filipp},\ and\ \citenamefont {Wallraff}}]{Yves}%
  \BibitemOpen
  \bibfield  {author} {\bibinfo {author} {\bibfnamefont {Y.}~\bibnamefont
  {Salath\'e}}, \bibinfo {author} {\bibfnamefont {M.}~\bibnamefont {Mondal}},
  \bibinfo {author} {\bibfnamefont {M.}~\bibnamefont {Oppliger}}, \bibinfo
  {author} {\bibfnamefont {J.}~\bibnamefont {Heinsoo}}, \bibinfo {author}
  {\bibfnamefont {P.}~\bibnamefont {Kurpiers}}, \bibinfo {author}
  {\bibfnamefont {A.}~\bibnamefont {Poto\ifmmode~\check{c}\else
  \v{c}\fi{}nik}}, \bibinfo {author} {\bibfnamefont {A.}~\bibnamefont
  {Mezzacapo}}, \bibinfo {author} {\bibfnamefont {U.}~\bibnamefont
  {Las~Heras}}, \bibinfo {author} {\bibfnamefont {L.}~\bibnamefont {Lamata}},
  \bibinfo {author} {\bibfnamefont {E.}~\bibnamefont {Solano}}, \bibinfo
  {author} {\bibfnamefont {S.}~\bibnamefont {Filipp}}, \ and\ \bibinfo {author}
  {\bibfnamefont {A.}~\bibnamefont {Wallraff}},\ }\bibinfo {title} {\emph
  {Digital quantum simulation of spin models with circuit quantum
  electrodynamics}},\ \href@noop {} {\bibfield  {journal} {\bibinfo  {journal}
  {Phys. Rev. X}\ }\textbf {\bibinfo {volume} {5}},\ \bibinfo {pages} {021027}
  (\bibinfo {year} {2015})}\BibitemShut {NoStop}%
\bibitem [{\citenamefont {Pfeuty}(1970)}]{Pfeuty70}%
  \BibitemOpen
  \bibfield  {author} {\bibinfo {author} {\bibfnamefont {P.}~\bibnamefont
  {Pfeuty}},\ }\bibinfo {title} {\emph {The one-dimensional Ising model with a
  transverse field}},\ \href@noop {} {\bibfield  {journal} {\bibinfo  {journal}
  {Ann. Phys.}\ }\textbf {\bibinfo {volume} {57}},\ \bibinfo {pages} {79}
  (\bibinfo {year} {1970})}\BibitemShut {NoStop}%
\bibitem [{\citenamefont {Sachdev}(2007)}]{sachdev2007quantum}%
  \BibitemOpen
  \bibfield  {author} {\bibinfo {author} {\bibfnamefont {S.}~\bibnamefont
  {Sachdev}},\ }\bibinfo {title} {\emph {Quantum phase transitions}},\
  \href@noop {} {\bibfield  {journal} {\bibinfo  {journal} {Handbook of
  Magnetism and Advanced Magnetic Materials}\ } (\bibinfo {year}
  {2007})}\BibitemShut {NoStop}%
\bibitem [{\citenamefont {Kinross}\ \emph {et~al.}(2014)\citenamefont
  {Kinross}, \citenamefont {Fu}, \citenamefont {Munsie}, \citenamefont
  {Dabkowska}, \citenamefont {Luke}, \citenamefont {Sachdev},\ and\
  \citenamefont {Imai}}]{PhysRevX.4.031008}%
  \BibitemOpen
  \bibfield  {author} {\bibinfo {author} {\bibfnamefont {A.~W.}\ \bibnamefont
  {Kinross}}, \bibinfo {author} {\bibfnamefont {M.}~\bibnamefont {Fu}},
  \bibinfo {author} {\bibfnamefont {T.~J.}\ \bibnamefont {Munsie}}, \bibinfo
  {author} {\bibfnamefont {H.~A.}\ \bibnamefont {Dabkowska}}, \bibinfo {author}
  {\bibfnamefont {G.~M.}\ \bibnamefont {Luke}}, \bibinfo {author}
  {\bibfnamefont {S.}~\bibnamefont {Sachdev}}, \ and\ \bibinfo {author}
  {\bibfnamefont {T.}~\bibnamefont {Imai}},\ }\bibinfo {title} {\emph
  {Evolution of quantum fluctuations near the quantum critical point of the
  transverse field Ising chain system ${\mathrm{CoNb}}_{2}{\mathrm{O}}_{6}$}},\
  \href@noop {} {\bibfield  {journal} {\bibinfo  {journal} {Phys. Rev. X}\
  }\textbf {\bibinfo {volume} {4}},\ \bibinfo {pages} {031008} (\bibinfo {year}
  {2014})}\BibitemShut {NoStop}%
\bibitem [{\citenamefont {di~Francesco}\ \emph {et~al.}(2012)\citenamefont
  {di~Francesco}, \citenamefont {Mathieu},\ and\ \citenamefont
  {S{\'e}n{\'e}chal}}]{francesco2012conformal}%
  \BibitemOpen
  \bibfield  {author} {\bibinfo {author} {\bibfnamefont {P.}~\bibnamefont
  {di~Francesco}}, \bibinfo {author} {\bibfnamefont {P.}~\bibnamefont
  {Mathieu}}, \ and\ \bibinfo {author} {\bibfnamefont {D.}~\bibnamefont
  {S{\'e}n{\'e}chal}},\ }\href@noop {} {\emph {\bibinfo {title} {Conformal
  field theory}}}\ (\bibinfo  {publisher} {Springer Science \& Business
  Media},\ \bibinfo {year} {2012})\BibitemShut {NoStop}%
\bibitem [{\citenamefont {Kitaev}(2001)}]{kitaev2001unpaired}%
  \BibitemOpen
  \bibfield  {author} {\bibinfo {author} {\bibfnamefont {A.~Y.}\ \bibnamefont
  {Kitaev}},\ }\bibinfo {title} {\emph {Unpaired Majorana fermions in quantum
  wires}},\ \href@noop {} {\bibfield  {journal} {\bibinfo  {journal}
  {Physics-Uspekhi}\ }\textbf {\bibinfo {volume} {44}},\ \bibinfo {pages} {131}
  (\bibinfo {year} {2001})}\BibitemShut {NoStop}%
\bibitem [{\citenamefont {You}\ and\ \citenamefont
  {Xu}(2014)}]{you2014symmetry}%
  \BibitemOpen
  \bibfield  {author} {\bibinfo {author} {\bibfnamefont {Y.-Z.}\ \bibnamefont
  {You}}\ and\ \bibinfo {author} {\bibfnamefont {C.}~\bibnamefont {Xu}},\
  }\bibinfo {title} {\emph {Symmetry-protected topological states of
  interacting fermions and bosons}},\ \href@noop {} {\bibfield  {journal}
  {\bibinfo  {journal} {Phys. Rev. B}\ }\textbf {\bibinfo {volume} {90}},\
  \bibinfo {pages} {245120} (\bibinfo {year} {2014})}\BibitemShut {NoStop}%
\bibitem [{\citenamefont {Katsura}\ \emph {et~al.}(2015)\citenamefont
  {Katsura}, \citenamefont {Schuricht},\ and\ \citenamefont
  {Takahashi}}]{PhysRevB.92.115137}%
  \BibitemOpen
  \bibfield  {author} {\bibinfo {author} {\bibfnamefont {H.}~\bibnamefont
  {Katsura}}, \bibinfo {author} {\bibfnamefont {D.}~\bibnamefont {Schuricht}},
  \ and\ \bibinfo {author} {\bibfnamefont {M.}~\bibnamefont {Takahashi}},\
  }\bibinfo {title} {\emph {Exact ground states and topological order in
  interacting Kitaev/Majorana chains}},\ \href@noop {} {\bibfield  {journal}
  {\bibinfo  {journal} {Phys. Rev. B}\ }\textbf {\bibinfo {volume} {92}},\
  \bibinfo {pages} {115137} (\bibinfo {year} {2015})}\BibitemShut {NoStop}%
\bibitem [{\citenamefont {Pollmann}\ \emph {et~al.}(2012)\citenamefont
  {Pollmann}, \citenamefont {Berg}, \citenamefont {Turner},\ and\ \citenamefont
  {Oshikawa}}]{pollmann2012symmetry}%
  \BibitemOpen
  \bibfield  {author} {\bibinfo {author} {\bibfnamefont {F.}~\bibnamefont
  {Pollmann}}, \bibinfo {author} {\bibfnamefont {E.}~\bibnamefont {Berg}},
  \bibinfo {author} {\bibfnamefont {A.~M.}\ \bibnamefont {Turner}}, \ and\
  \bibinfo {author} {\bibfnamefont {M.}~\bibnamefont {Oshikawa}},\ }\bibinfo
  {title} {\emph {Symmetry protection of topological phases in one-dimensional
  quantum spin systems}},\ \href@noop {} {\bibfield  {journal} {\bibinfo
  {journal} {Phys. Rev. B}\ }\textbf {\bibinfo {volume} {85}},\ \bibinfo
  {pages} {075125} (\bibinfo {year} {2012})}\BibitemShut {NoStop}%
\bibitem [{\citenamefont {Grimm}(2002)}]{grimm2002spectrum}%
  \BibitemOpen
  \bibfield  {author} {\bibinfo {author} {\bibfnamefont {U.}~\bibnamefont
  {Grimm}},\ }\bibinfo {title} {\emph {Spectrum of a duality-twisted Ising
  quantum chain}},\ \href@noop {} {\bibfield  {journal} {\bibinfo  {journal}
  {J. Phys. A}\ }\textbf {\bibinfo {volume} {35}},\ \bibinfo {pages} {L25}
  (\bibinfo {year} {2002})}\BibitemShut {NoStop}%
\bibitem [{\citenamefont {Nishimori}\ and\ \citenamefont
  {Takada}(2016)}]{nishimori2016exponential}%
  \BibitemOpen
  \bibfield  {author} {\bibinfo {author} {\bibfnamefont {H.}~\bibnamefont
  {Nishimori}}\ and\ \bibinfo {author} {\bibfnamefont {K.}~\bibnamefont
  {Takada}},\ }\bibinfo {title} {\emph {Exponential enhancement of the
  efficiency of quantum annealing by non-stochastic Hamiltonians}},\ \href@noop
  {} {\bibfield  {journal} {\bibinfo  {journal} {arXiv:1609.03785}\ } (\bibinfo
  {year} {2016})}\BibitemShut {NoStop}%
\bibitem [{\citenamefont {Calabrese}\ \emph
  {et~al.}(2012{\natexlab{a}})\citenamefont {Calabrese}, \citenamefont
  {Essler},\ and\ \citenamefont {Fagotti}}]{calabrese2012quantum}%
  \BibitemOpen
  \bibfield  {author} {\bibinfo {author} {\bibfnamefont {P.}~\bibnamefont
  {Calabrese}}, \bibinfo {author} {\bibfnamefont {F.~H.}\ \bibnamefont
  {Essler}}, \ and\ \bibinfo {author} {\bibfnamefont {M.}~\bibnamefont
  {Fagotti}},\ }\bibinfo {title} {\emph {Quantum quench in the transverse field
  Ising chain: I. Time evolution of order parameter correlators}},\ \href@noop
  {} {\bibfield  {journal} {\bibinfo  {journal} {J. Stat. Mech.}\ }\textbf
  {\bibinfo {volume} {2012}},\ \bibinfo {pages} {P07016} (\bibinfo {year}
  {2012}{\natexlab{a}})}\BibitemShut {NoStop}%
\bibitem [{\citenamefont {Calabrese}\ \emph
  {et~al.}(2012{\natexlab{b}})\citenamefont {Calabrese}, \citenamefont
  {Essler},\ and\ \citenamefont {Fagotti}}]{calabrese2012quantumII}%
  \BibitemOpen
  \bibfield  {author} {\bibinfo {author} {\bibfnamefont {P.}~\bibnamefont
  {Calabrese}}, \bibinfo {author} {\bibfnamefont {F.~H.}\ \bibnamefont
  {Essler}}, \ and\ \bibinfo {author} {\bibfnamefont {M.}~\bibnamefont
  {Fagotti}},\ }\bibinfo {title} {\emph {Quantum quenches in the transverse
  field Ising chain: II. Stationary state properties}},\ \href@noop {}
  {\bibfield  {journal} {\bibinfo  {journal} {J. Stat. Mech.}\ }\textbf
  {\bibinfo {volume} {2012}},\ \bibinfo {pages} {P07022} (\bibinfo {year}
  {2012}{\natexlab{b}})}\BibitemShut {NoStop}%
\bibitem [{\citenamefont {Calabrese}\ and\ \citenamefont
  {Cardy}(2004)}]{calabrese2004entanglement}%
  \BibitemOpen
  \bibfield  {author} {\bibinfo {author} {\bibfnamefont {P.}~\bibnamefont
  {Calabrese}}\ and\ \bibinfo {author} {\bibfnamefont {J.}~\bibnamefont
  {Cardy}},\ }\bibinfo {title} {\emph {Entanglement entropy and quantum field
  theory}},\ \href@noop {} {\bibfield  {journal} {\bibinfo  {journal} {J. Stat.
  Mech.}\ }\textbf {\bibinfo {volume} {2004}},\ \bibinfo {pages} {P06002}
  (\bibinfo {year} {2004})}\BibitemShut {NoStop}%
\bibitem [{\citenamefont {Dziarmaga}(2005)}]{dziarmaga2005dynamics}%
  \BibitemOpen
  \bibfield  {author} {\bibinfo {author} {\bibfnamefont {J.}~\bibnamefont
  {Dziarmaga}},\ }\bibinfo {title} {\emph {Dynamics of a quantum phase
  transition: Exact solution of the quantum Ising model}},\ \href@noop {}
  {\bibfield  {journal} {\bibinfo  {journal} {Phys. Rev. Lett.}\ }\textbf
  {\bibinfo {volume} {95}},\ \bibinfo {pages} {245701} (\bibinfo {year}
  {2005})}\BibitemShut {NoStop}%
\bibitem [{\citenamefont {Gluza}\ \emph {et~al.}(2016)\citenamefont {Gluza},
  \citenamefont {Krumnow}, \citenamefont {Friesdorf}, \citenamefont {Gogolin},\
  and\ \citenamefont {Eisert}}]{Gaussification}%
  \BibitemOpen
  \bibfield  {author} {\bibinfo {author} {\bibfnamefont {M.}~\bibnamefont
  {Gluza}}, \bibinfo {author} {\bibfnamefont {C.}~\bibnamefont {Krumnow}},
  \bibinfo {author} {\bibfnamefont {M.}~\bibnamefont {Friesdorf}}, \bibinfo
  {author} {\bibfnamefont {C.}~\bibnamefont {Gogolin}}, \ and\ \bibinfo
  {author} {\bibfnamefont {J.}~\bibnamefont {Eisert}},\ }\bibinfo {title}
  {\emph {Equilibration via Gaussification in fermionic lattice systems}},\
  \href@noop {} {\bibfield  {journal} {\bibinfo  {journal} {Phys. Rev. Lett.}\
  }\textbf {\bibinfo {volume} {117}},\ \bibinfo {pages} {190602} (\bibinfo
  {year} {2016})}\BibitemShut {NoStop}%
\bibitem [{\citenamefont {Schneider}\ \emph
  {et~al.}(2012{\natexlab{b}})\citenamefont {Schneider}, \citenamefont
  {Hackerm{\"u}ller}, \citenamefont {Ronzheimer}, \citenamefont {Will},
  \citenamefont {S.~Braun}, \citenamefont {Bloch}, \citenamefont {Demler},
  \citenamefont {Mandt}, \citenamefont {Rasch},\ and\ \citenamefont
  {Rosch}}]{RoschTransport}%
  \BibitemOpen
  \bibfield  {author} {\bibinfo {author} {\bibfnamefont {U.}~\bibnamefont
  {Schneider}}, \bibinfo {author} {\bibfnamefont {L.}~\bibnamefont
  {Hackerm{\"u}ller}}, \bibinfo {author} {\bibfnamefont {J.~P.}\ \bibnamefont
  {Ronzheimer}}, \bibinfo {author} {\bibfnamefont {S.}~\bibnamefont {Will}},
  \bibinfo {author} {\bibfnamefont {T.~B.}\ \bibnamefont {S.~Braun}}, \bibinfo
  {author} {\bibfnamefont {I.}~\bibnamefont {Bloch}}, \bibinfo {author}
  {\bibfnamefont {E.}~\bibnamefont {Demler}}, \bibinfo {author} {\bibfnamefont
  {S.}~\bibnamefont {Mandt}}, \bibinfo {author} {\bibfnamefont
  {D.}~\bibnamefont {Rasch}}, \ and\ \bibinfo {author} {\bibfnamefont
  {A.}~\bibnamefont {Rosch}},\ }\bibinfo {title} {\emph {Fermionic transport
  and out-of-equilibrium dynamics in a homogeneous Hubbard model with ultracold
  atoms}},\ \href@noop {} {\bibfield  {journal} {\bibinfo  {journal} {Nature
  Phys.}\ }\textbf {\bibinfo {volume} {8}},\ \bibinfo {pages} {213} (\bibinfo
  {year} {2012}{\natexlab{b}})}\BibitemShut {NoStop}%
\bibitem [{\citenamefont {Knill}\ \emph {et~al.}(2001)\citenamefont {Knill},
  \citenamefont {Laflamme},\ and\ \citenamefont
  {Milburn}}]{knill2000efficient}%
  \BibitemOpen
  \bibfield  {author} {\bibinfo {author} {\bibfnamefont {E.}~\bibnamefont
  {Knill}}, \bibinfo {author} {\bibfnamefont {R.}~\bibnamefont {Laflamme}}, \
  and\ \bibinfo {author} {\bibfnamefont {G.}~\bibnamefont {Milburn}},\
  }\bibinfo {title} {\emph {A scheme for efficient quantum computation with
  linear optics}},\ \href@noop {} {\bibfield  {journal} {\bibinfo  {journal}
  {Nature}\ }\textbf {\bibinfo {volume} {409}},\ \bibinfo {pages} {46}
  (\bibinfo {year} {2001})}\BibitemShut {NoStop}%
\bibitem [{\citenamefont {Knill}(2001)}]{knill2001fermionic}%
  \BibitemOpen
  \bibfield  {author} {\bibinfo {author} {\bibfnamefont {E.}~\bibnamefont
  {Knill}},\ }\bibinfo {title} {\emph {Fermionic linear optics and
  matchgates}},\ \href@noop {} {\bibfield  {journal} {\bibinfo  {journal}
  {quant-ph/0108033}\ } (\bibinfo {year} {2001})}\BibitemShut {NoStop}%
\bibitem [{\citenamefont {Terhal}\ and\ \citenamefont
  {DiVincenzo}(2002)}]{terhal2002classical}%
  \BibitemOpen
  \bibfield  {author} {\bibinfo {author} {\bibfnamefont {B.~M.}\ \bibnamefont
  {Terhal}}\ and\ \bibinfo {author} {\bibfnamefont {D.~P.}\ \bibnamefont
  {DiVincenzo}},\ }\bibinfo {title} {\emph {Classical simulation of
  non-interacting-fermion quantum circuits}},\ \href@noop {} {\bibfield
  {journal} {\bibinfo  {journal} {Phys. Rev. A}\ }\textbf {\bibinfo {volume}
  {65}},\ \bibinfo {pages} {032325} (\bibinfo {year} {2002})}\BibitemShut
  {NoStop}%
\bibitem [{\citenamefont {Bravyi}(2005)}]{bravyi2004lagrangian}%
  \BibitemOpen
  \bibfield  {author} {\bibinfo {author} {\bibfnamefont {S.}~\bibnamefont
  {Bravyi}},\ }\bibinfo {title} {\emph {Lagrangian representation for fermionic
  linear optics}},\ \href@noop {} {\bibfield  {journal} {\bibinfo  {journal}
  {Quantum Inf. and Comp.}\ }\textbf {\bibinfo {volume} {5}},\ \bibinfo {pages}
  {216} (\bibinfo {year} {2005})}\BibitemShut {NoStop}%
\bibitem [{\citenamefont {Kitaev}(2006)}]{kitaev2006anyons}%
  \BibitemOpen
  \bibfield  {author} {\bibinfo {author} {\bibfnamefont {A.}~\bibnamefont
  {Kitaev}},\ }\bibinfo {title} {\emph {Anyons in an exactly solved model and
  beyond}},\ \href@noop {} {\bibfield  {journal} {\bibinfo  {journal} {Ann.
  Phys.}\ }\textbf {\bibinfo {volume} {321}},\ \bibinfo {pages} {2} (\bibinfo
  {year} {2006})}\BibitemShut {NoStop}%
\bibitem [{\citenamefont {de~Melo}\ \emph {et~al.}(2013)\citenamefont
  {de~Melo}, \citenamefont {\'{C}wikli\'{n}ski},\ and\ \citenamefont
  {Terhal}}]{Melo13}%
  \BibitemOpen
  \bibfield  {author} {\bibinfo {author} {\bibfnamefont {F.}~\bibnamefont
  {de~Melo}}, \bibinfo {author} {\bibfnamefont {P.}~\bibnamefont
  {\'{C}wikli\'{n}ski}}, \ and\ \bibinfo {author} {\bibfnamefont {B.~M.}\
  \bibnamefont {Terhal}},\ }\bibinfo {title} {\emph {The power of noisy
  fermionic quantum computation}},\ \href@noop {} {\bibfield  {journal}
  {\bibinfo  {journal} {New J. Phys.}\ }\textbf {\bibinfo {volume} {15}},\
  \bibinfo {pages} {013015} (\bibinfo {year} {2013})}\BibitemShut {NoStop}%
\bibitem [{\citenamefont {{Bravyi}}(2005)}]{bravyicapacity}%
  \BibitemOpen
  \bibfield  {author} {\bibinfo {author} {\bibfnamefont {S.}~\bibnamefont
  {{Bravyi}}},\ }\bibinfo {title} {\emph {{Classical capacity of fermionic
  product channels}}},\ \href@noop {} {\bibfield  {journal} {\bibinfo
  {journal} {quant-ph/0507282}\ } (\bibinfo {year} {2005})}\BibitemShut
  {NoStop}%
\bibitem [{\citenamefont {Ferraro}\ \emph {et~al.}(2005)\citenamefont
  {Ferraro}, \citenamefont {Olivares},\ and\ \citenamefont {Paris}}]{Ferra05}%
  \BibitemOpen
  \bibfield  {author} {\bibinfo {author} {\bibfnamefont {A.}~\bibnamefont
  {Ferraro}}, \bibinfo {author} {\bibfnamefont {S.}~\bibnamefont {Olivares}}, \
  and\ \bibinfo {author} {\bibfnamefont {M.~G.~A.}\ \bibnamefont {Paris}},\
  }\bibinfo {title} {\emph {Gaussian states in continuous variable quantum
  information}},\ \href@noop {} {\bibfield  {journal} {\bibinfo  {journal}
  {arXiv:quant-ph/0503237}\ } (\bibinfo {year} {2005})}\BibitemShut {NoStop}%
\bibitem [{\citenamefont {Weedbrook}\ \emph {et~al.}(2012)\citenamefont
  {Weedbrook}, \citenamefont {Pirandola}, \citenamefont
  {Garc\'{i}a-Patr\'{o}n}, \citenamefont {Cerf}, \citenamefont {Ralph},
  \citenamefont {Shapiro},\ and\ \citenamefont {Lloyd}}]{Weed12}%
  \BibitemOpen
  \bibfield  {author} {\bibinfo {author} {\bibfnamefont {C.}~\bibnamefont
  {Weedbrook}}, \bibinfo {author} {\bibfnamefont {S.}~\bibnamefont
  {Pirandola}}, \bibinfo {author} {\bibfnamefont {R.}~\bibnamefont
  {Garc\'{i}a-Patr\'{o}n}}, \bibinfo {author} {\bibfnamefont {N.~J.}\
  \bibnamefont {Cerf}}, \bibinfo {author} {\bibfnamefont {T.~C.}\ \bibnamefont
  {Ralph}}, \bibinfo {author} {\bibfnamefont {J.~H.}\ \bibnamefont {Shapiro}},
  \ and\ \bibinfo {author} {\bibfnamefont {S.}~\bibnamefont {Lloyd}},\
  }\bibinfo {title} {\emph {Gaussian quantum information}},\ \href@noop {}
  {\bibfield  {journal} {\bibinfo  {journal} {Rev. Mod. Phys.}\ }\textbf
  {\bibinfo {volume} {84}},\ \bibinfo {pages} {621} (\bibinfo {year}
  {2012})}\BibitemShut {NoStop}%
\bibitem [{\citenamefont {Haldane}(1981)}]{haldane1981luttinger}%
  \BibitemOpen
  \bibfield  {author} {\bibinfo {author} {\bibfnamefont {F.}~\bibnamefont
  {Haldane}},\ }\bibinfo {title} {\emph {'Luttinger liquid theory'of
  one-dimensional quantum fluids. I. Properties of the Luttinger model and
  their extension to the general 1D interacting spinless Fermi gas}},\
  \href@noop {} {\bibfield  {journal} {\bibinfo  {journal} {J. Phys. C}\
  }\textbf {\bibinfo {volume} {14}},\ \bibinfo {pages} {2585} (\bibinfo {year}
  {1981})}\BibitemShut {NoStop}%
\bibitem [{\citenamefont {von Delft~Jan}\ and\ \citenamefont
  {Herbert}(1998)}]{vonDelft_bosonization}%
  \BibitemOpen
  \bibfield  {author} {\bibinfo {author} {\bibnamefont {von Delft~Jan}}\ and\
  \bibinfo {author} {\bibfnamefont {S.}~\bibnamefont {Herbert}},\ }\bibinfo
  {title} {\emph {Bosonization for beginners — refermionization for
  experts}},\ \href@noop {} {\bibfield  {journal} {\bibinfo  {journal} {Ann.
  Phys.}\ }\textbf {\bibinfo {volume} {7}},\ \bibinfo {pages} {225} (\bibinfo
  {year} {1998})}\BibitemShut {NoStop}%
\bibitem [{\citenamefont {Fradkin}(2013)}]{fradkin2013field}%
  \BibitemOpen
  \bibfield  {author} {\bibinfo {author} {\bibfnamefont {E.}~\bibnamefont
  {Fradkin}},\ }\href@noop {} {\emph {\bibinfo {title} {Field theories of
  condensed matter physics}}}\ (\bibinfo  {publisher} {Cambridge University
  Press},\ \bibinfo {year} {2013})\BibitemShut {NoStop}%
\bibitem [{\citenamefont {Gross}\ \emph {et~al.}(2010)\citenamefont {Gross},
  \citenamefont {Liu}, \citenamefont {Flammia}, \citenamefont {Becker},\ and\
  \citenamefont {Eisert}}]{Gross10}%
  \BibitemOpen
  \bibfield  {author} {\bibinfo {author} {\bibfnamefont {D.}~\bibnamefont
  {Gross}}, \bibinfo {author} {\bibfnamefont {Y.-K.}\ \bibnamefont {Liu}},
  \bibinfo {author} {\bibfnamefont {S.~T.}\ \bibnamefont {Flammia}}, \bibinfo
  {author} {\bibfnamefont {S.}~\bibnamefont {Becker}}, \ and\ \bibinfo {author}
  {\bibfnamefont {J.}~\bibnamefont {Eisert}},\ }\bibinfo {title} {\emph
  {Quantum state tomography via compressed sensing}},\ \href@noop {} {\bibfield
   {journal} {\bibinfo  {journal} {Phys. Rev. Lett.}\ }\textbf {\bibinfo
  {volume} {105}},\ \bibinfo {pages} {150401} (\bibinfo {year}
  {2010})}\BibitemShut {NoStop}%
\bibitem [{\citenamefont {Riofrio}\ \emph {et~al.}(2017)\citenamefont
  {Riofrio}, \citenamefont {Gross}, \citenamefont {Flammia}, \citenamefont
  {Monz}, \citenamefont {Nigg}, \citenamefont {Blatt},\ and\ \citenamefont
  {Eisert}}]{ExperimentalCS}%
  \BibitemOpen
  \bibfield  {author} {\bibinfo {author} {\bibfnamefont {C.~A.}\ \bibnamefont
  {Riofrio}}, \bibinfo {author} {\bibfnamefont {D.}~\bibnamefont {Gross}},
  \bibinfo {author} {\bibfnamefont {S.~T.}\ \bibnamefont {Flammia}}, \bibinfo
  {author} {\bibfnamefont {T.}~\bibnamefont {Monz}}, \bibinfo {author}
  {\bibfnamefont {D.}~\bibnamefont {Nigg}}, \bibinfo {author} {\bibfnamefont
  {R.}~\bibnamefont {Blatt}}, \ and\ \bibinfo {author} {\bibfnamefont
  {J.}~\bibnamefont {Eisert}},\ }\bibinfo {title} {\emph {Experimental quantum
  compressed sensing for a seven-qubit system}},\ \href@noop {} {\bibfield
  {journal} {\bibinfo  {journal} {Nature Comm.}\ }\textbf {\bibinfo {volume}
  {8}},\ \bibinfo {pages} {15305} (\bibinfo {year} {2017})}\BibitemShut
  {NoStop}%
\bibitem [{\citenamefont {T\'{o}th}\ \emph {et~al.}(2010)\citenamefont
  {T\'{o}th}, \citenamefont {Wieczorek}, \citenamefont {Gross}, \citenamefont
  {Krischek}, \citenamefont {Schwemmer},\ and\ \citenamefont
  {Weinfurter}}]{Toth10}%
  \BibitemOpen
  \bibfield  {author} {\bibinfo {author} {\bibfnamefont {G.}~\bibnamefont
  {T\'{o}th}}, \bibinfo {author} {\bibfnamefont {W.}~\bibnamefont {Wieczorek}},
  \bibinfo {author} {\bibfnamefont {D.}~\bibnamefont {Gross}}, \bibinfo
  {author} {\bibfnamefont {R.}~\bibnamefont {Krischek}}, \bibinfo {author}
  {\bibfnamefont {C.}~\bibnamefont {Schwemmer}}, \ and\ \bibinfo {author}
  {\bibfnamefont {H.}~\bibnamefont {Weinfurter}},\ }\bibinfo {title} {\emph
  {Permutationally invariant quantum tomography}},\ \href@noop {} {\bibfield
  {journal} {\bibinfo  {journal} {Phys. Rev. Lett.}\ }\textbf {\bibinfo
  {volume} {105}},\ \bibinfo {pages} {250403} (\bibinfo {year}
  {2010})}\BibitemShut {NoStop}%
\bibitem [{\citenamefont {Cramer}\ \emph {et~al.}(2010)\citenamefont {Cramer},
  \citenamefont {Plenio}, \citenamefont {Flammia}, \citenamefont {Somma},
  \citenamefont {Gross}, \citenamefont {Bartlett}, \citenamefont
  {Landon-Cardinal}, \citenamefont {Poulin},\ and\ \citenamefont
  {Liu}}]{cramer2010efficient}%
  \BibitemOpen
  \bibfield  {author} {\bibinfo {author} {\bibfnamefont {M.}~\bibnamefont
  {Cramer}}, \bibinfo {author} {\bibfnamefont {M.~B.}\ \bibnamefont {Plenio}},
  \bibinfo {author} {\bibfnamefont {S.~T.}\ \bibnamefont {Flammia}}, \bibinfo
  {author} {\bibfnamefont {R.}~\bibnamefont {Somma}}, \bibinfo {author}
  {\bibfnamefont {D.}~\bibnamefont {Gross}}, \bibinfo {author} {\bibfnamefont
  {S.~D.}\ \bibnamefont {Bartlett}}, \bibinfo {author} {\bibfnamefont
  {O.}~\bibnamefont {Landon-Cardinal}}, \bibinfo {author} {\bibfnamefont
  {D.}~\bibnamefont {Poulin}}, \ and\ \bibinfo {author} {\bibfnamefont {Y.-K.}\
  \bibnamefont {Liu}},\ }\bibinfo {title} {\emph {Efficient quantum state
  tomography}},\ \href@noop {} {\bibfield  {journal} {\bibinfo  {journal}
  {Nature Comm.}\ }\textbf {\bibinfo {volume} {1}},\ \bibinfo {pages} {149}
  (\bibinfo {year} {2010})}\BibitemShut {NoStop}%
\bibitem [{\citenamefont {Flammia}\ and\ \citenamefont
  {Liu}(2011)}]{flammia2011direct}%
  \BibitemOpen
  \bibfield  {author} {\bibinfo {author} {\bibfnamefont {S.~T.}\ \bibnamefont
  {Flammia}}\ and\ \bibinfo {author} {\bibfnamefont {Y.-K.}\ \bibnamefont
  {Liu}},\ }\bibinfo {title} {\emph {Direct fidelity estimation from few Pauli
  measurements}},\ \href@noop {} {\bibfield  {journal} {\bibinfo  {journal}
  {Phys. Rev. Lett.}\ }\textbf {\bibinfo {volume} {106}},\ \bibinfo {pages}
  {230501} (\bibinfo {year} {2011})}\BibitemShut {NoStop}%
\bibitem [{\citenamefont {da~Silva}\ \emph {et~al.}(2011)\citenamefont
  {da~Silva}, \citenamefont {Landon-Cardinal},\ and\ \citenamefont
  {Poulin}}]{SilLanPou11}%
  \BibitemOpen
  \bibfield  {author} {\bibinfo {author} {\bibfnamefont {M.~P.}\ \bibnamefont
  {da~Silva}}, \bibinfo {author} {\bibfnamefont {O.}~\bibnamefont
  {Landon-Cardinal}}, \ and\ \bibinfo {author} {\bibfnamefont {D.}~\bibnamefont
  {Poulin}},\ }\bibinfo {title} {\emph {Practical characterization of quantum
  devices without tomography}},\ \href@noop {} {\bibfield  {journal} {\bibinfo
  {journal} {Phys. Rev. Lett.}\ }\textbf {\bibinfo {volume} {107}},\ \bibinfo
  {pages} {210404} (\bibinfo {year} {2011})}\BibitemShut {NoStop}%
\bibitem [{\citenamefont {Flammia}\ \emph {et~al.}(2014)\citenamefont
  {Flammia}, \citenamefont {Gross}, \citenamefont {Liu},\ and\ \citenamefont
  {Eisert}}]{Flammia12}%
  \BibitemOpen
  \bibfield  {author} {\bibinfo {author} {\bibfnamefont {S.~T.}\ \bibnamefont
  {Flammia}}, \bibinfo {author} {\bibfnamefont {D.}~\bibnamefont {Gross}},
  \bibinfo {author} {\bibfnamefont {Y.-K.}\ \bibnamefont {Liu}}, \ and\
  \bibinfo {author} {\bibfnamefont {J.}~\bibnamefont {Eisert}},\ }\bibinfo
  {title} {\emph {Quantum tomography via compressed sensing: {E}rror bounds,
  sample complexity, and efficient estimators}},\ \href@noop {} {\bibfield
  {journal} {\bibinfo  {journal} {New J. Phys.}\ }\textbf {\bibinfo {volume}
  {14}},\ \bibinfo {pages} {095022} (\bibinfo {year} {2014})}\BibitemShut
  {NoStop}%
\bibitem [{\citenamefont {Aolita}\ \emph {et~al.}(2015)\citenamefont {Aolita},
  \citenamefont {Gogolin}, \citenamefont {Kliesch},\ and\ \citenamefont
  {Eisert}}]{aolita2015reliable}%
  \BibitemOpen
  \bibfield  {author} {\bibinfo {author} {\bibfnamefont {L.}~\bibnamefont
  {Aolita}}, \bibinfo {author} {\bibfnamefont {C.}~\bibnamefont {Gogolin}},
  \bibinfo {author} {\bibfnamefont {M.}~\bibnamefont {Kliesch}}, \ and\
  \bibinfo {author} {\bibfnamefont {J.}~\bibnamefont {Eisert}},\ }\bibinfo
  {title} {\emph {Reliable quantum certification of photonic state
  preparations}},\ \href@noop {} {\bibfield  {journal} {\bibinfo  {journal}
  {Nature Comm.}\ }\textbf {\bibinfo {volume} {6}},\ \bibinfo {pages} {8498}
  (\bibinfo {year} {2015})}\BibitemShut {NoStop}%
\bibitem [{\citenamefont {Steffens}\ \emph {et~al.}(2015)\citenamefont
  {Steffens}, \citenamefont {Friesdorf}, \citenamefont {Langen}, \citenamefont
  {Rauer}, \citenamefont {Schweigler}, \citenamefont {H{\"u}bener},
  \citenamefont {Schmiedmayer}, \citenamefont {Riofr{\'\i}o},\ and\
  \citenamefont {Eisert}}]{steffens2015towards}%
  \BibitemOpen
  \bibfield  {author} {\bibinfo {author} {\bibfnamefont {A.}~\bibnamefont
  {Steffens}}, \bibinfo {author} {\bibfnamefont {M.}~\bibnamefont {Friesdorf}},
  \bibinfo {author} {\bibfnamefont {T.}~\bibnamefont {Langen}}, \bibinfo
  {author} {\bibfnamefont {B.}~\bibnamefont {Rauer}}, \bibinfo {author}
  {\bibfnamefont {T.}~\bibnamefont {Schweigler}}, \bibinfo {author}
  {\bibfnamefont {R.}~\bibnamefont {H{\"u}bener}}, \bibinfo {author}
  {\bibfnamefont {J.}~\bibnamefont {Schmiedmayer}}, \bibinfo {author}
  {\bibfnamefont {C.}~\bibnamefont {Riofr{\'\i}o}}, \ and\ \bibinfo {author}
  {\bibfnamefont {J.}~\bibnamefont {Eisert}},\ }\bibinfo {title} {\emph
  {Towards experimental quantum-field tomography with ultracold atoms}},\
  \href@noop {} {\bibfield  {journal} {\bibinfo  {journal} {Nature Comm.}\
  }\textbf {\bibinfo {volume} {6}} (\bibinfo {year} {2015})}\BibitemShut
  {NoStop}%
\bibitem [{\citenamefont {{Hangleiter}}\ \emph {et~al.}(2017)\citenamefont
  {{Hangleiter}}, \citenamefont {{Kliesch}}, \citenamefont {{Schwarz}},\ and\
  \citenamefont {{Eisert}}}]{Hangleiter16}%
  \BibitemOpen
  \bibfield  {author} {\bibinfo {author} {\bibfnamefont {D.}~\bibnamefont
  {{Hangleiter}}}, \bibinfo {author} {\bibfnamefont {M.}~\bibnamefont
  {{Kliesch}}}, \bibinfo {author} {\bibfnamefont {M.}~\bibnamefont
  {{Schwarz}}}, \ and\ \bibinfo {author} {\bibfnamefont {J.}~\bibnamefont
  {{Eisert}}},\ }\bibinfo {title} {\emph {Direct certification of a class of
  quantum simulations}},\ \href@noop {} {\bibfield  {journal} {\bibinfo
  {journal} {Quantum Sci. Technol.}\ }\textbf {\bibinfo {volume} {2}},\
  \bibinfo {pages} {015004} (\bibinfo {year} {2017})}\BibitemShut {NoStop}%
\bibitem [{\citenamefont {Eisert}\ and\ \citenamefont
  {Osborne}(2006)}]{eisert2006general}%
  \BibitemOpen
  \bibfield  {author} {\bibinfo {author} {\bibfnamefont {J.}~\bibnamefont
  {Eisert}}\ and\ \bibinfo {author} {\bibfnamefont {T.~J.}\ \bibnamefont
  {Osborne}},\ }\bibinfo {title} {\emph {General entanglement scaling laws from
  time evolution}},\ \href@noop {} {\bibfield  {journal} {\bibinfo  {journal}
  {Phys. Rev. Lett.}\ }\textbf {\bibinfo {volume} {97}},\ \bibinfo {pages}
  {150404} (\bibinfo {year} {2006})}\BibitemShut {NoStop}%
\bibitem [{\citenamefont {Ram{\'\i}rez}\ \emph {et~al.}(2015)\citenamefont
  {Ram{\'\i}rez}, \citenamefont {Rodr{\'\i}guez-Laguna},\ and\ \citenamefont
  {Sierra}}]{ramirez2015entanglement}%
  \BibitemOpen
  \bibfield  {author} {\bibinfo {author} {\bibfnamefont {G.}~\bibnamefont
  {Ram{\'\i}rez}}, \bibinfo {author} {\bibfnamefont {J.}~\bibnamefont
  {Rodr{\'\i}guez-Laguna}}, \ and\ \bibinfo {author} {\bibfnamefont
  {G.}~\bibnamefont {Sierra}},\ }\bibinfo {title} {\emph {Entanglement over the
  rainbow}},\ \href@noop {} {\bibfield  {journal} {\bibinfo  {journal} {J.
  Stat. Mech.}\ }\textbf {\bibinfo {volume} {2015}},\ \bibinfo {pages} {P06002}
  (\bibinfo {year} {2015})}\BibitemShut {NoStop}%
\bibitem [{\citenamefont {Eisler}\ and\ \citenamefont
  {Zimbor{\'a}s}(2016)}]{eisler2016entanglement}%
  \BibitemOpen
  \bibfield  {author} {\bibinfo {author} {\bibfnamefont {V.}~\bibnamefont
  {Eisler}}\ and\ \bibinfo {author} {\bibfnamefont {Z.}~\bibnamefont
  {Zimbor{\'a}s}},\ }\bibinfo {title} {\emph {Entanglement negativity in
  two-dimensional free lattice models}},\ \href@noop {} {\bibfield  {journal}
  {\bibinfo  {journal} {Phys. Rev. B}\ }\textbf {\bibinfo {volume} {93}},\
  \bibinfo {pages} {115148} (\bibinfo {year} {2016})}\BibitemShut {NoStop}%
\bibitem [{\citenamefont {Kraus}\ and\ \citenamefont
  {Cirac}(2010)}]{kraus2010generalized}%
  \BibitemOpen
  \bibfield  {author} {\bibinfo {author} {\bibfnamefont {C.~V.}\ \bibnamefont
  {Kraus}}\ and\ \bibinfo {author} {\bibfnamefont {J.~I.}\ \bibnamefont
  {Cirac}},\ }\bibinfo {title} {\emph {Generalized Hartree--Fock theory for
  interacting fermions in lattices: numerical methods}},\ \href@noop {}
  {\bibfield  {journal} {\bibinfo  {journal} {New J. Phys.}\ }\textbf {\bibinfo
  {volume} {12}},\ \bibinfo {pages} {113004} (\bibinfo {year}
  {2010})}\BibitemShut {NoStop}%
\bibitem [{\citenamefont {G{\"{u}}hne}\ and\ \citenamefont
  {T{\'{o}}th}(2009)}]{guehne09}%
  \BibitemOpen
  \bibfield  {author} {\bibinfo {author} {\bibfnamefont {O.}~\bibnamefont
  {G{\"{u}}hne}}\ and\ \bibinfo {author} {\bibfnamefont {G.}~\bibnamefont
  {T{\'{o}}th}},\ }\bibinfo {title} {\emph {Entanglement detection}},\
  \href@noop {} {\bibfield  {journal} {\bibinfo  {journal} {Phys. Rep.}\
  }\textbf {\bibinfo {volume} {474}},\ \bibinfo {pages} {1} (\bibinfo {year}
  {2009})}\BibitemShut {NoStop}%
\bibitem [{\citenamefont {Fr{\"o}wis}\ \emph {et~al.}(2013)\citenamefont
  {Fr{\"o}wis}, \citenamefont {van~den Nest},\ and\ \citenamefont
  {D{\"u}r}}]{vanderNest}%
  \BibitemOpen
  \bibfield  {author} {\bibinfo {author} {\bibfnamefont {F.}~\bibnamefont
  {Fr{\"o}wis}}, \bibinfo {author} {\bibfnamefont {M.}~\bibnamefont {van~den
  Nest}}, \ and\ \bibinfo {author} {\bibfnamefont {W.}~\bibnamefont
  {D{\"u}r}},\ }\bibinfo {title} {\emph {Certifiability criterion for
  large-scale quantum system}},\ \href
  {http://stacks.iop.org/1367-2630/15/i=11/a=113011} {\bibfield  {journal}
  {\bibinfo  {journal} {New J. Phys.}\ }\textbf {\bibinfo {volume} {15}},\
  \bibinfo {pages} {113011} (\bibinfo {year} {2013})}\BibitemShut {NoStop}%
\bibitem [{\citenamefont {Gogolin}\ \emph {et~al.}(2013)\citenamefont
  {Gogolin}, \citenamefont {Kliesch}, \citenamefont {Aolita},\ and\
  \citenamefont {Eisert}}]{GogKliAol13}%
  \BibitemOpen
  \bibfield  {author} {\bibinfo {author} {\bibfnamefont {C.}~\bibnamefont
  {Gogolin}}, \bibinfo {author} {\bibfnamefont {M.}~\bibnamefont {Kliesch}},
  \bibinfo {author} {\bibfnamefont {L.}~\bibnamefont {Aolita}}, \ and\ \bibinfo
  {author} {\bibfnamefont {J.}~\bibnamefont {Eisert}},\ }\bibinfo {title}
  {\emph {Boson-sampling in the light of sample complexity}},\ \href@noop {}
  {\bibfield  {journal} {\bibinfo  {journal} {arxiv e-prints
  ArXiv:1306.3995v2}\ } (\bibinfo {year} {2013})}\BibitemShut {NoStop}%
\bibitem [{\citenamefont {Hastings}\ and\ \citenamefont
  {Koma}(2006)}]{hastings2006spectral}%
  \BibitemOpen
  \bibfield  {author} {\bibinfo {author} {\bibfnamefont {M.~B.}\ \bibnamefont
  {Hastings}}\ and\ \bibinfo {author} {\bibfnamefont {T.}~\bibnamefont
  {Koma}},\ }\bibinfo {title} {\emph {Spectral gap and exponential decay of
  correlations}},\ \href@noop {} {\bibfield  {journal} {\bibinfo  {journal}
  {Commun. Math. Phys.}\ }\textbf {\bibinfo {volume} {265}},\ \bibinfo {pages}
  {781} (\bibinfo {year} {2006})}\BibitemShut {NoStop}%
\bibitem [{\citenamefont {Jordan}\ and\ \citenamefont
  {Eugene}(1928)}]{JordanWigner}%
  \BibitemOpen
  \bibfield  {author} {\bibinfo {author} {\bibfnamefont {P.}~\bibnamefont
  {Jordan}}\ and\ \bibinfo {author} {\bibfnamefont {W.}~\bibnamefont
  {Eugene}},\ }\bibinfo {title} {\emph {{\"U}ber das Paulische
  {\"A}quivalenzverbot}},\ \href@noop {} {\bibfield  {journal} {\bibinfo
  {journal} {Z. Physik}\ }\textbf {\bibinfo {volume} {47}},\ \bibinfo {pages}
  {631} (\bibinfo {year} {1928})}\BibitemShut {NoStop}%
\bibitem [{\citenamefont {Lieb}\ \emph {et~al.}(1961)\citenamefont {Lieb},
  \citenamefont {Schultz},\ and\ \citenamefont {Mattis}}]{LiebSchultzMattis}%
  \BibitemOpen
  \bibfield  {author} {\bibinfo {author} {\bibfnamefont {E.}~\bibnamefont
  {Lieb}}, \bibinfo {author} {\bibfnamefont {T.}~\bibnamefont {Schultz}}, \
  and\ \bibinfo {author} {\bibfnamefont {D.}~\bibnamefont {Mattis}},\ }\bibinfo
  {title} {\emph {Two soluble models of an antiferromagnetic chain}},\
  \href@noop {} {\bibfield  {journal} {\bibinfo  {journal} {Ann. Phys.}\
  }\textbf {\bibinfo {volume} {16}},\ \bibinfo {pages} {407 } (\bibinfo {year}
  {1961})}\BibitemShut {NoStop}%
\bibitem [{\citenamefont {Blatt}\ and\ \citenamefont
  {Roos}(2012{\natexlab{b}})}]{BlattRoos12}%
  \BibitemOpen
  \bibfield  {author} {\bibinfo {author} {\bibfnamefont {R.}~\bibnamefont
  {Blatt}}\ and\ \bibinfo {author} {\bibfnamefont {C.}~\bibnamefont {Roos}},\
  }\bibinfo {title} {\emph {Quantum simulations with trapped ions}},\
  \href@noop {} {\bibfield  {journal} {\bibinfo  {journal} {Nature Phys.}\
  }\textbf {\bibinfo {volume} {8}},\ \bibinfo {pages} {277} (\bibinfo {year}
  {2012}{\natexlab{b}})}\BibitemShut {NoStop}%
\bibitem [{\citenamefont {Schmied}\ \emph {et~al.}(2011)\citenamefont
  {Schmied}, \citenamefont {Wesenberg},\ and\ \citenamefont
  {Leibfried}}]{schmied2011quantum}%
  \BibitemOpen
  \bibfield  {author} {\bibinfo {author} {\bibfnamefont {R.}~\bibnamefont
  {Schmied}}, \bibinfo {author} {\bibfnamefont {J.~H.}\ \bibnamefont
  {Wesenberg}}, \ and\ \bibinfo {author} {\bibfnamefont {D.}~\bibnamefont
  {Leibfried}},\ }\bibinfo {title} {\emph {Quantum simulation of the hexagonal
  Kitaev model with trapped ions}},\ \href@noop {} {\bibfield  {journal}
  {\bibinfo  {journal} {New J. Phys.}\ }\textbf {\bibinfo {volume} {13}},\
  \bibinfo {pages} {115011} (\bibinfo {year} {2011})}\BibitemShut {NoStop}%
\bibitem [{\citenamefont {Mielenz}\ \emph {et~al.}(2015)\citenamefont
  {Mielenz}, \citenamefont {Kalis}, \citenamefont {Wittemer}, \citenamefont
  {Hakelberg}, \citenamefont {Schmied}, \citenamefont {Blain}, \citenamefont
  {Maunz}, \citenamefont {Leibfried}, \citenamefont {Warring},\ and\
  \citenamefont {Schaetz}}]{mielenz2015freely}%
  \BibitemOpen
  \bibfield  {author} {\bibinfo {author} {\bibfnamefont {M.}~\bibnamefont
  {Mielenz}}, \bibinfo {author} {\bibfnamefont {H.}~\bibnamefont {Kalis}},
  \bibinfo {author} {\bibfnamefont {M.}~\bibnamefont {Wittemer}}, \bibinfo
  {author} {\bibfnamefont {F.}~\bibnamefont {Hakelberg}}, \bibinfo {author}
  {\bibfnamefont {R.}~\bibnamefont {Schmied}}, \bibinfo {author} {\bibfnamefont
  {M.}~\bibnamefont {Blain}}, \bibinfo {author} {\bibfnamefont
  {P.}~\bibnamefont {Maunz}}, \bibinfo {author} {\bibfnamefont
  {D.}~\bibnamefont {Leibfried}}, \bibinfo {author} {\bibfnamefont
  {U.}~\bibnamefont {Warring}}, \ and\ \bibinfo {author} {\bibfnamefont
  {T.}~\bibnamefont {Schaetz}},\ }\bibinfo {title} {\emph {Freely configurable
  quantum simulator based on a two-dimensional array of individually trapped
  ions}},\ \href@noop {} {\bibfield  {journal} {\bibinfo  {journal}
  {arXiv:1512.03559}\ } (\bibinfo {year} {2015})}\BibitemShut {NoStop}%
\bibitem [{\citenamefont {Gluza}(2018)}]{github}%
  \BibitemOpen
  \bibfield  {author} {\bibinfo {author} {\bibfnamefont {M.}~\bibnamefont
  {Gluza}},\ }\href@noop {} {}\bibinfo {howpublished}
  {\url{https://github.com/marekgluza/Fidelity_witnesses_example}} (\bibinfo
  {year} {2018})\BibitemShut {NoStop}%
\bibitem [{\citenamefont {Wimmer}(2012)}]{PFAPACK}%
  \BibitemOpen
  \bibfield  {author} {\bibinfo {author} {\bibfnamefont {M.}~\bibnamefont
  {Wimmer}},\ }\bibinfo {title} {\emph {Algorithm 923: Efficient numerical
  computation of the Pfaffian for dense and banded skew-symmetric matrices}},\
  \href {\doibase 10.1145/2331130.2331138} {\bibfield  {journal} {\bibinfo
  {journal} {ACM Trans. Math. Softw.}\ }\textbf {\bibinfo {volume} {38}},\
  \bibinfo {pages} {30:1} (\bibinfo {year} {2012})}\BibitemShut {NoStop}%
\end{thebibliography}
\end{document}